\documentclass{article}
\usepackage{aastex-macros}
\usepackage[margin=1in]{geometry}
\usepackage[utf8]{inputenc}
\usepackage{amssymb,amsmath,amsthm}
\usepackage{xcolor}
\usepackage{graphicx}
\usepackage{braket}
\usepackage{caption}
\usepackage{subcaption}
\usepackage{hyperref}
\usepackage{tikz}
\usetikzlibrary{quantikz}
\usetikzlibrary{positioning}
\usepackage{adjustbox}
\usepackage{multirow}
\usepackage[inline]{enumitem}

\newtheorem{theorem}{Theorem}
\numberwithin{theorem}{section}
\newtheorem{proposition}[theorem]{Proposition}
\newtheorem{definition}[theorem]{Definition}
\newtheorem{lemma}[theorem]{Lemma}

\title{Error mitigation via error-detection using Generalized Superfast Encodings}
\author{Hagge, Tobias$^{1}$ \\
    \texttt{tobias.hagge@pnnl.gov}
    \and
    Wiebe, Nathan$^{1,2,3}$ \\
    \texttt{nawiebe@cs.toronto.edu}
}
\date{$^1$ Pacific Northwest National Laboratory, Richland, WA, USA\\
$^2$ University of Toronto, Toronto, ON, Canada\\
    $^3$ Canadian Institute for Advanced Studies, Toronto, ON, Canada}

\begin{document}

\maketitle

\begin{abstract}
    We provide a new approach to error mitigation for quantum chemistry simulation that uses a Bravyi--Kitaev Superfast encoding to implement a quantum error detecting code within the fermionic encoding.  Our construction has low-weight parity checks as well. We show that for the spinless Hubbard model with nearest-neighbor repulsion terms, one-qubit errors are detectable, and more complicated errors are detectable with high probability. While our error-detection requires additional quantum circuitry, we argue that there is a regime in which the beneficial effect of error-mitigation outweighs the deleterious effects of additional errors due to additional circuitry. We show that our scheme can be implemented under realistic qubit connectivity requirements.
\end{abstract}

\section{Introduction}

In the current noisy-intermediate scale quantum (NISQ) era, quantum resources for a computation are severely constrained. Quantum error mitigation \cite{2022arXiv221000921C} techniques attempt to improve the accuracy of results of quantum programs while using no or minimal additional quantum resources. One approach to error mitigation is to leverage physical symmetries inherent in the problem description or its representation \cite{PhysRevA.98.062339, 2021Quant...5..548C, bksfgse} to detect and correct errors. For electronic structure problems, one source for these symmetries is the encoding of fermions in qubit space; the Bravyi-Kitaev Superfast (BKSF) encoding \cite{bravyi-kitaev}, and variants such as the Generalized Superfast Encoding (GSE) \cite{bksfgse} and Majorana Loop Stabilizer Encoding (MLSE) \cite{2019PhRvP..12f4041J} have been shown to possess error-detecting properties, the latter two being capable of full one-qubit error-correction in some applications.

Unfortunately, such encodings only protect the encoded fermions, not the quantum circuits which manipulate these fermions. It turns out that the set of operations which are protected by the encoding does not even include evolutions in the code space. In the VQE experiment we will consider in this paper, the region of the quantum circuit to which the proven error-mitigating properties apply, without special effort or circumstances, has circuit depth one. Ensuring fault-tolerance properties for the rest of the circuit requires careful management of encoding choices and implementation, with ancilla qubits and extra circuitry in some cases. The limitations are reminiscent of gate-set limitations imposed by the Eastin-Knill theorem \cite{eastin-knill}, though that theorem does not apply here because the fermion-to-qubit encodings in question are not transverse.

Error-correction capabilities are further limited because, unlike qubit-to-qubit encodings, fermion-to-qubit encodings cannot be applied recursively to increase code distance. On the other hand, the error-detection properties of such codes are fairly resilient; multiple errors are highly likely to trigger an error detection even when code distance bounds are exceeded.

These limitations beg the question, is there a regime in which the error-mitigation properties of fermion-to-qubit encodings provide an advantage in practice? In this paper, we develop the use of stabilizer code properties of fermion to qubit mappings as an error mitigation technique. To this end, we describe error-detecting quantum circuits suitable for use in a VQE algorithm, and realize them under explicit and plausible qubit connectivity assumptions. We show that there is a plausible regime in which error-detection improves the accuracy of computed expectation values, at the cost of additional sampling complexity.

The content of the rest of this paper is as follows. In Section~\ref{sec:lfoa} we recall the definition of an edge and vertex algebra. In Section~\ref{sec:bksf} we review the fermion-to-qubit encoding in which we will work, the Bravyi-Kitaev generalized superfast encoding (GSE), motivate the choice of this particular fermion-to-qubit encoding in our context, and describe a procedure for zero-state initialization. In Section~\ref{sec:hubbard} we recall the definition of the Fermi-Hubbard model and describe an error-detecting variant for spinless two-dimensional lattices. In Section~\ref{sec:fault-detecting} we develop circuitry sufficient to implement a VQE algorithm using the GSE encoding, and analyze its requirements. In Section~\ref{sec:reduced connectivity} we develop a version of the results in Section~\ref{sec:fault-detecting} under reduced connectivity requirements. In Section~\ref{sec:performance} we analyze the performance of fault-detecting VQE circuits under reduced connectivity requirements.

\section{Edge and vertex algebras}\label{sec:lfoa}

To evolve fermionic quantum Hamiltonians with a quantum computer, it is necessary to represent them in qubit space. For second-quantized Hamiltonians relevant to quantum computing applications, the Hamiltonian $H$ is commonly represented using the fermionic operator algebra generated by creation and annihilation operators $a_j^\dagger$ and $a_j$, respectively, indexed over $m$ fermionic occupancy sites $j \in 1 \ldots m$. The creation and annihilation operators have the following relations:

\begin{enumerate}
    \item $a_j a_k^\dagger+ a_k^\dagger a_j = \delta_{j,k}$,
    \item $a_j a_k + a_k a_j = 0$,
    \item $a_j^\dagger a_k^\dagger + a_k^\dagger a_j^\dagger = 0$.
\end{enumerate}

The presence of commutation and anti-commutation relations makes it natural to represent $a_j$ and $a_j^\dagger$ in qubit space with Pauli operators. The most straightforward approach is the Jordan-Wigner embedding. Low-weight Pauli operators, however, have few anti-commutation relations relative to the number possessed by raising and lowering operators. This dichotomy constrains the compactness of the encoding; $O(m)$ qubits are required to represent each $a_j$ and $a_j^\dagger$ in the Jordon-Wigner formalism.

To solve this problem, a family of fermion-to-qubit mappings attempt to lower the weight of fermionic operators which appear in a given Hamiltonian. Some of these mappings \cite{2005JSMTE..09..012V, 2016PhRvA..94c0301W, 2019PhRvA..99b2308S} attempt to modify the Jordan-Wigner representation directly. Others \cite{bravyi-kitaev, bksfgse, 2019PhRvP..12f4041J, 2021PhRvB.104c5118D, 2023PRXQ....4a0326C}, including the original Bravyi-Kitaev superfast method and the generalized superfast encoding method is based, reduce the weight of qubit operators by working in the even fermionic operator subalgebra and encoding a different basis of operators in qubit space. All of these methods can be interpreted as variants of exact bosonization
\cite{2023PRXQ....4a0326C}.

For any fermionic Hamiltonian $H$, expressed as a creation and annihilation operator polynomial, there is an interaction graph $G$ with one vertex for each site index $1 \ldots m$ in $H$ and one edge for each pair of site indices which co-occur in an $H$ summand. The summands are contained in a fermionic operator subalgebra $A_G$, the {\em local fermionic operator algebra}, which is generated by edge and vertex operators $A_{j,k}$ and $B_j$, defined below, one for each edge $(j,k)$ and each vertex $j$ respectively of $G$ \cite{bravyi-kitaev}.

\begin{definition}\label{def:edge and vertex operators}
The edge and vertex operators are usually defined in terms of the Majorana operators $c_{2j}$ and $c_{2j+1}$, for $j \in 1 \ldots m$, which are as follows:

\begin{align*}c_{2j} := a_j + a_j^\dagger\qquad
    c_{2j+1} := \frac{a_j - a_j^\dagger}i,\qquad
    \end{align*}
The edge and vertex operators are then:
\begin{align*}
    B_j = -i c_{2j} c_{2j+1} = I - 2 a_j^\dagger a_j, \qquad A_{j,k} = -i c_{2j}c_{2k} = -i (a_j + a_j^\dagger)(a_k + a_k^\dagger).
\end{align*}
\end{definition}

\begin{proposition}\label{prop:useful equalities}
The following algebraic properties hold.
\begin{enumerate}
\item 
for all $j,k \in 1 \ldots 2m$, the Majorana operators satisfy the following:
\begin{align}
    c_j c_k + c_k c_j = 2 \delta_{j,k}.
\end{align}

\item
For all edges $(j,k)$ in the interaction graph $G$, the edge and vertex operators $A_{jk}$, $B_j$, and $B_k$ satisfy the following:
\begin{align}
    A_{j,k}^\dagger &= A_{j,k},\\
    B_j^\dagger &= B_j,\\
    A_{j,k}^2 &= B_j^2 = 1,\\
    B_j B_k &= B_k B_j,\\
    A_{j,k} &= - A_{k,j},\\
    A_{j,k} A_{k,l} &= - A_{k,l} A_{j,k}, \text{if $j \ne l$},\\
    A_{j,k} A_{l,m} &= A_{l,m} A_{j,k} \text{ if $j,k,l,m$ are distinct},\\
    A_{j,k} B_j &= - B_j A_{j,k},\\
    A_{j,k} B_l &= B_l A_{j,k} \text{ if $j,k,l$ are distinct},\\
    i^n\prod_{j=0}^{n-1}A_{k_j,k_{(j+1 \mod n)}} &= 1, \text{for each cycle with ordered vertices $(k_0, \ldots k_{n-1})$ in $G$}. \label{loop operator relation}
\end{align}

\item The following equalities relate raising and lowering operators to edge and vertex operators:
\begin{align}
    \frac{I - B_j}{2} &= a_j^\dagger a_j,\\
    B_j A_{j,k} &= c_{2j+1}c_{2k} = -i (a_j - a_j^\dagger)(a_k + a_k^\dagger),\\
    A_{j,k} B_k &= -c_{2j}c_{2k+1} = i (a_j + a_j^\dagger)(a_k - a_k^\dagger),\\
    B_j A_{j,k} + A_{j,k} B_k &= 2i(a_j^\dagger a_k + a_k^\dagger a_j)
\end{align}
\end{enumerate}   
\end{proposition}
\begin{proof}
All of these properties follow algebraically from the definitions.
\end{proof}

The above relations show that each edge and vertex operator anticommutes with those edge and vertex operators with which it shares a single vertex. The edge and vertex qubit operators we will consider have weight $O(d)$, where $d$ is the degree of the graph $G$. This is an asymptotic improvement over the $O(n)$ weight required for Jordan-Wigner operators for families of graphs with bounded $d$.

\section{The Bravyi-Kitaev superfast encoding, generalized superfast encodings (GSE), and stabilizer codes}\label{sec:bksf}

\subsection{Generalized superfast encodings}
The superfast family of encodings map edge and vertex operators to Pauli operators, in a way that satisfies all of the edge and vertex operator relations except Equation~\ref{loop operator relation}. The left hand side of Equation~\ref{loop operator relation} is known as a {\em loop operator} and plays a special role in the construction. To make Equation~\ref{loop operator relation} hold, fermionic states must be constrained during initial state preparation to lie in the mutual $+1$ eigenspace of all loop operators.

The generalized superfast encoding (GSE) construction of \cite{bksfgse} improves on the original \cite{bravyi-kitaev} construction by reducing operator weight and providing error-correction properties. The authors observe that because fermionic states lie within the mutual $+1$ eigenspaces of commuting Pauli loop operators, the loop operators are stabilizers in a quantum stabilizer code \cite{gottesman-thesis}.

Stabilizer codes allow correction of errors in quantum computations \cite{1996quant.ph..5011S}, and when a threshold of gate accuracy is present and sufficient quantum resources are available, enable fault-tolerant quantum computing \cite{1996quant.ph..5011S,1996quant.ph.11025A,2009arXiv0904.2557G}. Currently, the most promising stabilizer codes for this purpose are surface codes \cite{1998quant.ph.11052B}, due to their low-weight operators, modest qubit connectivity requirements and generous error thresholds \cite{2002JMP....43.4452D}. Conceptually, surface-code qubits are connected in two-dimensional planar lattice configurations; prominent current quantum hardware schemes have qubit connectivity which efficiently supports computations on lattices \cite{48651,PhysRevX.10.011022}.

A stabilizer code can correct arbitrary one-qubit errors if and only if every logical operator has weight at least three. The authors of \cite{bksfgse} show that for the generalized superfast encoding, under mild connectivity conditions and degree at least six for every vertex, the edge and vertex operators may be chosen so that all logical operators have weight three or greater, and thus the encoding corrects arbitrary one-qubit errors. This choice results in $O(d)$ edge and vertex operator weight.

The GSE encoding requires that the Hamiltonian interaction graph $G = (V,E)$ be of even degree $d(v)$ at each $v \in V$. If need be, this assumption can be satisfied by augmenting $G$ with additional edges representing zero-amplitude Hamiltonian interaction terms.

\begin{figure}
    \centering
    \begin{tikzpicture}
        \draw[color=red] (0,1) -- (1,1) node [midway,above] {$\gamma_{3,1}$};
        \draw[color=red] (2,1) -- (1,1) node [midway,below] {$\gamma_{3,3}$};
        \draw[color=red] (1,0) -- (1,1) node [midway,left] {$\gamma_{3,4}$};
        \draw[color=red] (1,2) -- (1,1) node [midway,right] {$\gamma_{3,2}$};
        \draw[color=blue] (2,1) -- (3,1) node [midway,above] {$\gamma_{4,1}$};
        \draw[color=blue] (4,1) -- (3,1) node [midway,below] {$\gamma_{4,3}$};
        \draw[color=blue] (3,0) -- (3,1) node [midway,left] {$\gamma_{4,4}$};
        \draw[color=blue] (3,2) -- (3,1) node [midway,right] {$\gamma_{4,2}$};
        \draw[color=green] (0,3) -- (1,3) node [midway,above] {$\gamma_{1,1}$};
        \draw[color=green] (2,3) -- (1,3) node [midway,below] {$\gamma_{1,3}$};
        \draw[color=green] (1,2) -- (1,3) node [midway,left] {$\gamma_{1,4}$};
        \draw[color=green] (1,4) -- (1,3) node [midway,right] {$\gamma_{1,2}$};
        \draw[color=orange] (2,3) -- (3,3) node [midway,above] {$\gamma_{2,1}$};
        \draw[color=orange] (4,3) -- (3,3) node [midway,below] {$\gamma_{2,3}$};
        \draw[color=orange] (3,2) -- (3,3) node [midway,left] {$\gamma_{2,4}$};
        \draw[color=orange] (3,4) -- (3,3) node [midway,right] {$\gamma_{2,2}$};
        
    \end{tikzpicture}
    \caption{To construct local fermionic algebra operators, generalized Majorana operators $\gamma_{v,j}$ for vertex $v$ must be assigned to the outgoing edges of $v$. One assignment on a square lattice (degree four) is shown here. Pairs of $\gamma_{i,j}$ anticommute if they lie on half-edges of the same color, and commute otherwise.}
    \label{fig:anticommuting gamma}
\end{figure}

The construction of the GSE encoding of~\cite{bksfgse} is as follows. To each $v \in V$, assign $\frac{d(v)}{2}$ qubits. Choose $d(v)$ mutually anticommuting Pauli operators $\gamma_{v,1} \ldots \gamma_{v,d(v)}$ with support on those qubits, assigning one to each half-edge incident to $v$.

For example, Figure~\ref{fig:anticommuting gamma} shows the mapping of $\gamma_{v,j}$ to half-edges which will later be used for the square lattice, and Figure~\ref{fig:spinless hubbard weight} shows the Pauli operators to which these $\gamma_{v,j}$ will map.

To define the edge and vertex operators, for each edge $\{j,k\}$ choose an orientation $\epsilon_{j,k} \in \pm 1$, $\epsilon_{j,k} = - \epsilon_{k,j}$. Then define

\begin{align}
    A_{j,k} & = \epsilon_{j,k} \gamma_{j,m_{j,k}} \gamma_{k,n_{j,k}}, \\
    B_j & = (-i)^{\frac{d(j)}2} \prod_{m=1}^{d(j)} \gamma_{j,m},
\end{align}

where $\gamma_{j,m_{j,k}}$ and $\gamma_{k,n_{j,k}}$ are the half-edge operators corresponding to edge $\{j,k\}$.

The operators $\gamma_{v,i}$ are called {\em generalized Majorana operators}. The anticommutativity relations of generalized Majorana operators are leveraged to construct representations of $A_{jk}$ and $B_j$ with correct commutativity properties. The operators are ``generalized'' in the sense of having similar formal properties; there is no mapping or correspondence between individual Majorana operators used to construct $A_{jk}$ and $B_j$ and the generalized operators $\gamma_{v,i}$ introduced in the next step.

\subsection{Generalized superfast encodings as stabilizer codes}
Under these definitions the loop operators defined in \ref{loop operator relation} correspond to Pauli operators. Following the stabilizer code formalism, these Pauli operators generate the loop operator subgroup $\mathcal L$ which is a subgroup of the Pauli group $\mathcal{G}$. Taken as group elements, loop operators are not independent; elements of $\mathcal L$ are generated up to sign by a smaller set of basis loop operators. In the case of an embedded planar graph the plaquette loop operators form a basis; more precisely they form a basis for the interaction graph's first homology group $H_1$ with $\mathbb Z_2$ coefficients.

The code space for the resulting stabilizer code is defined as the quotient group  $C_\mathcal{G}(\mathcal L)/\mathcal L$, that is the group of elements $C_\mathcal{G}(\mathcal L)$ that commute with all loop operators, modulo the loop operator subgroup $\mathcal L$. The logical operators in the code are defined to be the nontrivial elements of $C_\mathcal{G}(\mathcal L)/\mathcal L$, or equivalently the elements of $C_\mathcal{G}(\mathcal L) - \mathcal L$, taken as equivalence class representatives.

By construction, the $A_{j,k}$ and $B_j$ operators lie in $C_\mathcal{G}(\mathcal L)$. By a dimension-counting argument (see \cite{bksfgse}), they generate $C_\mathcal{G}(\mathcal L)$, but in general some may be trivial, and they may not all lie in distinct cosets (and represent logical operators). In practice, however, neither of these issues present difficulties. The only nontrivial edge or vertex operators are the self-loops, which do not appear in the Hamiltonian interaction graph construction. Furthermore, operators can share a coset only when they form a doubled-edge loop. In this case, in the construction one of the operators has weight zero.

Exploiting the error-correcting properties of the GSE presents difficulties in practice. Without modification of the encoding, degree six is lowest degree for which single-qubit error-correction is possible, as the weight of a $B_j$ operator is at most twice the degree of the vertex, and single-qubit error correction requires that all logical operators have weight at least three. Vertices of degree six lead to significant qubit-connectivity requirements, which, if not satisfied by hardware, will be compensated for with qubit swap operators, increasing the edge and vertex operator weights. Finally, the code distance is determined by the degree of the Hamiltonian graph, the structure of its loop operators, and the choice of $\gamma_{v,j}$. Being a fermion-to-qubit mapping, recursive error-correction scaling methods for qubit-to-qubit mappings cannot be directly applied to increase the effective code distance.



\begin{figure}
    \centering
    \begin{tabular}{c|c|c|c| c}
         & Distance & Occupation ($B_j$) & Hopping ($A_{jk}$) & Stabilizer \\
         \hline
         BKSF & 2 & 4 & 6 & 6 \\
         MLSE & 3 & 3 & 3-4 & 4-10 \\
         GSE & 2 & 2 & 3-4 & 6 \\
    \end{tabular}
    \caption{Code distances and operator weights for BKSF, MLSE, and GSE. First two rows taken from \cite{2019PhRvP..12f4041J}.}
    \label{fig:distances and weights}
\end{figure}

\subsection{Definite-occupancy state preparation by syndrome measurement}
To perform computations with a stabilizer code it is necessary to produce an encoded initial state. Here, we demonstrate an efficient method for preparing an initial definite-occupancy state, applicable to GSE and other edge and vertex algebra encodings, based on the quiescient state method of \cite{fowler_2012}.

The desired initial state is an encoded state which
\begin{enumerate}
    \item has the correct orbital occupancies as measured by the $B_j$ operators,
    \item lies in the $+1$ eigenspace of all loop operators.
\end{enumerate}

The process is as follows. First, we prepare a state with the correct $B_j$ eigenvalues. Since each $B_j$ is a Pauli operator acting on the qubits assigned to the $j$-th site orbital, and the assignments are all disjoint, we can prepare a mutual $\pm 1$ eigenstate for all $B_j$ operators with a depth-one circuit.

Next, we measure the plaquette loop operators. These measurements do not alter subsequent $B_j$ operator measurements because loop operators commute with $B_j$ operators. The plaquette measurements collapse the state to a mutual loop operator eigenstate, since the plaquette operators generate the remaining loop operators, but usually not the mutual $+1$ eigenstate that is needed. Undesired $-1$ measurement values can be corrected without additional quantum circuitry by a change of operator basis, as shown in the following proposition: 

\begin{proposition}\label{prop change signs}
In the local edge-and-vertex algebra for the graph $G$ representing a square or toroidal lattice, suppose the plaquette loop operators $P$ are measured, and the set $D$  produce some defective (eigenvalue $-1$) syndrome measurements. The measured state is the $+1$ mutual-eigenstate for an edge-and-vertex algebra given by replacing some of the edge operators $A_{jk}$ with $-A_{jk}$.
\end{proposition}

\begin{proof}[Proof of Proposition~\ref{prop change signs}]

Let $\psi$ be the state that results from measurement, and let $D$ be the set of plaquette or lattice boundary operators with defective measurements. Then $|D|$ is even, since the boundary operator measurement is the product of the plaquette operator measurements. Choose a pair of loop operators $d_1, d_2$ in $D$ and a sequence $d_1 =p_1, p_2, \ldots, p_k = d_k$ of pairwise edge-adjacent loop operators, such that $p_i$ and $p_{i+1}$ share edge operator $e_i$. Replace $e_1, \ldots, e_{k-1}$ in the loop operator algebra with $-e_1, \ldots, -e_{k-1}$ respectively. In the new algebra, the defective loop operator set $D' = D - \{d_1, d_2\}$. Repeat until $D = \emptyset$.
\end{proof}

In the language of topology, we have reversed the signs on a $1$-cochain of edges, the co-boundary of which is the set of plaquettes with defective loop operator measurements.

\section{Realizing the Hubbard model}\label{sec:hubbard}
Here we describe the Hubbard model on the two-dimensional planar and toroidal lattices, along with its spinless variant. We express their Hamiltonians in the edge-and-vertex algebra, and construct a GSE encoding for the spinless case.

Let $G = (V,E)$ be the $M \times N$ grid graph, which is the direct graph product of path graphs $P_M$ and $P_N$ of size $M$ and $N$ respectively.
Then 
\[V = \{(i,j) | 0 \le i \le M - 1, 0 \le j \le N-1\}, \]
\[E = \{((i,j),(i+1,j)) |0 \le i \le M-2, 0 \le j \le N-1\} \bigcup \{((i,j),(i,j+1)) | 0 \le i \le M-1, 0 \le j \le N - 2\}. \]

The $M=N=4$ case is shown in the first diagram in Figure~\ref{fig:hubbard lattices}.

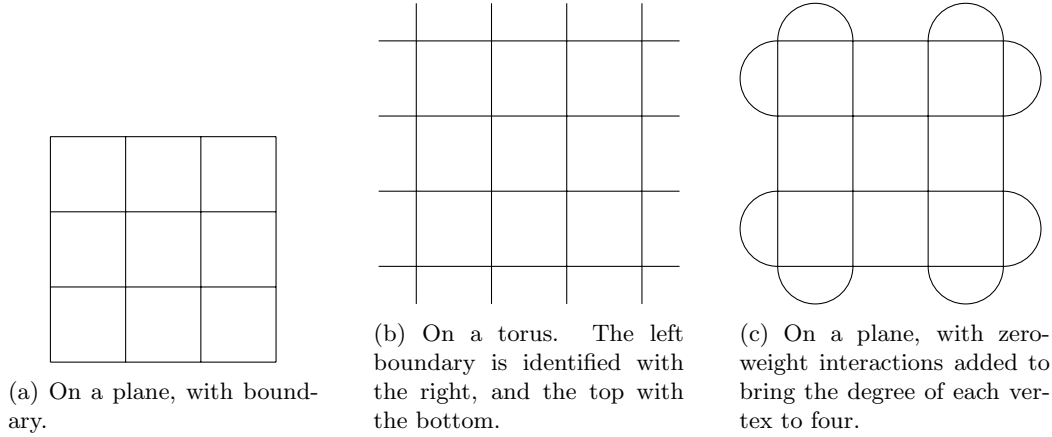
\begin{figure}
    \begin{subfigure}[b]{\linewidth/4}
        \centering
        \begin{tikzpicture}
            \draw (0,0) grid (3,3);
        \end{tikzpicture}
        \caption{On a plane, with boundary.}
    \end{subfigure}
    \hspace{.5cm}
    \begin{subfigure}[b]{\linewidth/4}
        \centering
        \begin{tikzpicture}
            \draw (0,0) grid (3,3);
            \draw (0,0) -- (-.5,0);
            \draw (0,1) -- (-.5,1);
            \draw (0,2) -- (-.5,2);
            \draw (0,3) -- (-.5,3);
            \draw (3,0) -- (3.5,0);
            \draw (3,1) -- (3.5,1);
            \draw (3,2) -- (3.5,2);
            \draw (3,3) -- (3.5,3);
            \draw (0,0) -- (0,-.5);        
            \draw (1,0) -- (1,-.5);        
            \draw (2,0) -- (2,-.5);        
            \draw (3,0) -- (3,-.5);
            \draw (0,3) -- (0,3.5);
            \draw (1,3) -- (1,3.5);
            \draw (2,3) -- (2,3.5);
            \draw (3,3) -- (3,3.5);
        \end{tikzpicture}
        \caption{On a torus. The left boundary is identified with the right, and the top with the bottom.}
        \label{fig:hubbard lattices:torus}
    \end{subfigure}
    \hspace{.5cm}
    \begin{subfigure}[b]{\linewidth/4}
        \begin{tikzpicture}
            \draw (0,0) grid (3,3);
            \draw (0,0) arc (180:360:.5);
            \draw (2,0) arc (180:360:.5);
            \draw (0,0) arc (270:90:.5);
            \draw (0,2) arc (270:90:.5);
            \draw (0,3) arc (180:0:.5);
            \draw (2,3) arc (180:0:.5);
            \draw (3,0) arc (-90:90:.5);
            \draw (3,2) arc (-90:90:.5);
        \end{tikzpicture}
        \caption{On a plane, with zero-weight interactions added to bring the degree of each vertex to four.}
        \label{fig:hubbard lattices:planar even}
    \end{subfigure}
    \caption{ Variants of the $4 \times 4$ Hubbard lattice. Each vertex represents a pair of spin orbitals, or a single spin orbital in the spinless case. Edges represent nontrivial Hamiltonian interactions among the orbitals. Adjacent-site, same spin interactions, along with same-site different-spin interactions, may be performed locally in the fermionic operator algebra using $A_{j,k}$ and $B_j$ operators not involving any other spin orbital.}
    \label{fig:hubbard lattices}
\end{figure}
The Hubbard Hamiltonian is as follows:
\[H = \sum_{(j,k) \in E, \sigma \in \{\uparrow,\downarrow\}} -t(a_{j,\sigma}^\dagger a_{k,\sigma} + a_{k,\sigma}^\dagger a_{j,\sigma}) + U \sum_{j \in V} a_{j,\uparrow}^\dagger a_{j, \uparrow} a_{j,\downarrow}^\dagger a_{j,\downarrow}. \]

Using the equalities in Proposition~\ref{prop:useful equalities}, the Hubbard Hamiltonian can be rewritten in terms of the edge and vertex operators as
\[H_H = \frac{i t}2 \sum_{(j,k) \in E, \sigma \in \{\uparrow, \downarrow\}} (B_{(j,\sigma)} A_{(j,\sigma),(k,\sigma)} + A_{(j,\sigma),(k,\sigma)}B_{(k,\sigma)}) + U \sum_{j \in V} \frac{(1-B_{(j,\uparrow)})(1-B_{(j,\downarrow)})} 4\]

For the spinless version, all fermions are constrained to have the same spin. A major difference between the two models is that for the Hubbard model with spin the Pauli exclusion principle permits two electrons with opposite spin to interact when at the same site.  For the spinless case, the Pauli exclusion principle forbids such interactions, and the system becomes a system of non-interacting fermions which can be exactly solved using classical resources.  To avoid this situation, the spinless model is assumed to include off-site nearest-neighbor repulsion terms in the Hamiltonian: 
\[H_\uparrow = \sum_{(j,k) \in E} -t(a_{j,\uparrow}^\dagger a_{k,\uparrow} + a_{k,\uparrow}^\dagger a_{j,\uparrow}) + U \sum_{j,k \in E} a_{j,\uparrow}^\dagger a_{j, \uparrow} a_{k,\uparrow}^\dagger a_{k,\uparrow},\]

This Hamiltonian is expressed in the edge-and-vertex algebra as follows:
\[H_{SLH} = \frac{i t}2 \sum_{(j,k) \in E} (B_{(j,\uparrow)} A_{(j,\uparrow),(k,\uparrow)} + A_{(j,\uparrow),(k,\uparrow)}B_{(k,\uparrow)}) + U \sum_{(j,k) \in E} \frac{(1-B_{(j,\uparrow)})(1-B_{(k,\uparrow)})} 4\]

Replacing the graph $G$ with another graph gives a Hamiltonian with the same summation, but summed over different edges. Figure~\ref{fig:hubbard lattices} shows the toroidal lattice, as well as a planar lattice curved arcs added to bring the degree of each vertex to four. In the second case, the curved edges are assumed to have interaction weight zero, and $M$ and $N$ to be even. Let the Hamiltonians obtained be denoted $H_{SLH,T}$ and $H_{SLH,4}$ respectively.

We work with the following GSE encoding of $H_{SLH,T}$ and $H_{SLH,4}$

\begin{definition}\label{def:our gse}
Let $G$ be a toroidal Hubbard lattice of size at least $2 \times 2$, or a planar Hubbard lattice with even numbers of rows and columns, and doubled edges as shown in Figure~\ref{fig:hubbard lattices}.

For each vertex $j$ of $G$, assign to $\gamma_{j,1}, \ldots, \gamma_{j,4}$ the values $(XY, YY, IX, IZ)$ respectively, using the mapping of $\gamma_{j,k}$ to half-edges shown in Figure~\ref{fig:anticommuting gamma}, resulting in the assignment of qubit operators to half-edges shown in Figure~\ref{fig:spinless hubbard weight}.\footnote{Generalized Majorana operators for degree four graphs are not given in \cite{bksfgse}, but there is an obvious way to extend conventions given for larger-degree graphs to the degree four case, which does give the conventions used in this paper. In the conventions of \cite{bksfgse}, some errors which can propagate from single-qubit errors during swap gates are not detectable; our choices eliminate this possibility.} The resulting  interior loop operator $IYXZYXZI$ has weight six, with weight three boundary loop operators $-YXZI$, $-IYYX$, $-IYXZ$, and $-XZZI$.
\end{definition}

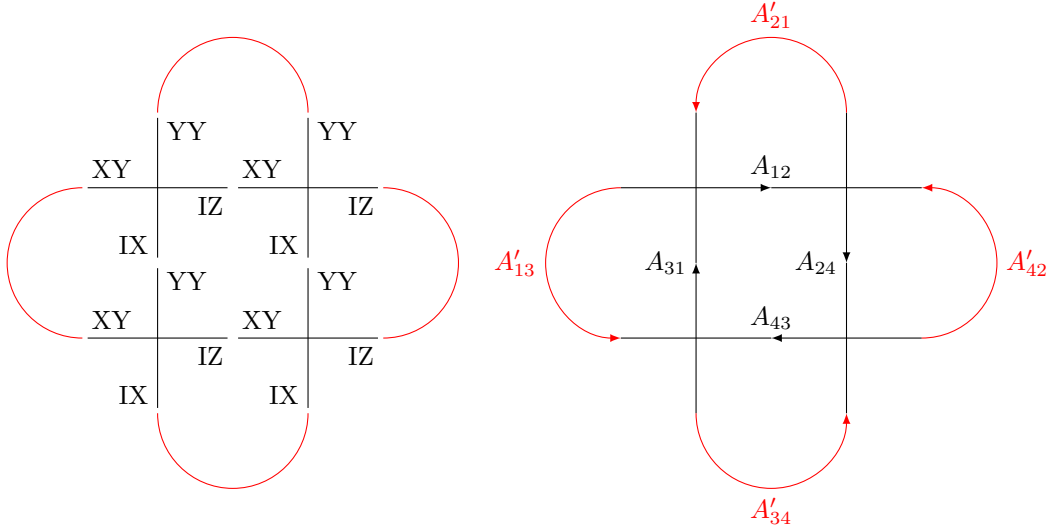
\begin{figure}[t!]
\captionsetup{singlelinecheck=off}
\centering


\raisebox{-.5\height}{
\begin{tikzpicture}
\draw (0,1)  node[above right] (l4) {XY} -- (2,1) node[below left] (r4) {IZ} node[above right] (l3) {XY}
    -- (4,1) node[below left] (r3) {IZ};
\draw (0,3)  node[above right] (l4) {XY} -- (2,3) node[below left] (r4) {IZ} node[above right] (l3) {XY}
    -- (4,3) node[below left] (r3) {IZ};
\draw (1,0) node[above left] (b4) {IX} -- (1,2) node[below right] (t4) {YY} node[above left] (b1) {IX} -- (1,4) node[below right] (t1) {YY};
\draw (3,0) node[above left] (b4) {IX} -- (3,2) node[below right] (t4) {YY} node[above left] (b1) {IX} -- (3,4) node[below right] (t1) {YY};
\draw[color=white, line width=4] (2,0) -- (2,4);
\draw[color=white, line width=4] (0,0) -- (0,4);
\draw[color=white, line width=4] (4,0) -- (4,4);
\draw[color=white, line width=4] (0,0) -- (4,0);
\draw[color=white, line width=4] (0,2) -- (4,2);
\draw[color=white, line width=4] (0,4) -- (4,4);


\draw[color=red] (0,1) arc (270:90:1);
\draw[color=red] (1,4) arc (180:0:1);
\draw[color=red] (4,3) arc (90:-90:1);
\draw[color=red] (3,0) arc (360:180:1);

\end{tikzpicture}
}
\raisebox{-.5\height}{
\begin{tikzpicture}
\draw (0,1) -- (2,1); 
\draw [latex-] (2,1) node[above] {$A_{43}$} -- (4,1);
\draw [-latex] (0,3) -- (2,3) node[above] {$A_{12}$};
\draw (2,3) -- (4,3);
\draw [-latex] (1,0) -- (1,2) node[left] {$A_{31}$};
\draw (1,2) -- (1,4);
\draw (3,0) -- (3,2);
\draw [latex-] (3,2) node[left] {$A_{24}$}-- (3,4);


\draw[color=red,latex-] (0,1) arc (270:180:1) node[left] {$A'_{13}$} arc (180:90:1);
\draw[color=red,latex-] (1,4) arc (180:90:1) node[above] {$A'_{21}$} arc (90:0:1);
\draw[color=red,latex-] (4,3) arc (90:0:1) node[right] {$A'_{42}$} arc (0:-90:1);
\draw[color=red,latex-] (3,0) arc (360:270:1) node[below] {$A'_{34}$} arc (270:180:1);

\end{tikzpicture}
}
\caption[Spinless Hubbard local qubit assignments]{GSE qubit operator assignments for half-edges on spinless Hubbard Hamiltonian interaction graphs. These choices produce a weight-six loop operator for spinless Hubbard without boundary, and weight-three loop operators for boundary loops. The qubit operators representing the generalized Majorana operators from Figure~\ref{fig:anticommuting gamma} are illustrated in the first image, with doubled edges for the boundary cases of the planar lattice indicated in red. Each operator acts on the pair of qubits assigned to its adjacent vertex. The orientation for each $A_{jk}$ operator has positive $\epsilon_{jk}$ when the orientation aligns with the arrow in the second image. From this image, one may compute that the horizontal $A_{jk}$ operators are of the form $\pm IZXY$, and the vertical, $\pm IXYY$. The $A_jk$ corresponding to the red doubled edges are of the forms $\pm XYXY$, $\pm YYYY$, $\pm IZIZ$, and $\pm IXIX$. The loop operator for the central plaquette, with qubits ordered by vertex as in Figure~\ref{fig:anticommuting gamma}, is $i^4 A_{12} A_{24} A_{43} A_{31} = IYXZYXZI$. The loop operators for boundary bigons are $i^2 A_{12} A'_{21} = -YXZI$, $i^2 A_{24} A'_{42} = -IYYX$, $i^2 A_{43} A'_{34} = -IYXZ$, and $i^2 A_{31} A'_{13}= -XZZI$.
  }
  \label{fig:spinless hubbard weight}
\end{figure}
Since there is no canonical ordering of sites, in contrast to the Jordan-Wigner encoding, ordering conventions for qubit-space representations of edge and vertex operators are needed. The ordering conventions of Figure~\ref{fig:spinless hubbard weight} will be used in the remainder of this paper: horizontal edges vertex-ordered from left to right, vertical edges up to down,loop operators as shown in Figure~\ref{fig:spinless hubbard weight}. We assume arbitrary qubit connectivity until Section~\ref{sec:reduced connectivity}, where a reduced connectivity implementation will be described.
\begin{proposition}
The choices given in Definition~\ref{def:our gse} produce a single-qubit error-detecting code.
\end{proposition}
\begin{proof}
Single-qubit Pauli errors are, by definition, weight-one Pauli operators. Detectable errors are Pauli operators which are not logical operators; these fail to commute with at least one syndrome measurement and thus produce an error detection. For single-qubit error-detection we must therefore show that every logical operator has weight at least two. In our code, every logical operator with support on two or more vertices has weight at least two. Suppose $R$ is a logical operator supported on a single vertex $j$. Then $R$ lies in the two-qubit Pauli group generated by generalized Majorana operators $\gamma_{j,k}$. Figure~\ref{fig:anticommuting gamma} shows the assignment of $\gamma_{j,k}$ to portions of plaquette loop operators at a vertex. For each such plaquette loop operator $L$ incident to vertex $j$, $R$ must contain either both of the $\gamma_{j,k}$ in $L$ incident to vertex $j$, or neither of them. Since this must hold for all four incident plaquette loops (or all three incident plaquette loops on a planar-lattice boundary vertex), either $R = I$, or $R = \pm B_j$. Since $B_j$ has weight two in our code, all logical operators have weight at least two. Thus single-qubit errors are detectable.
\end{proof}

\section{Fault-detecting circuitry for Hamiltonian simulation}\label{sec:fault-detecting}
As the BKSF encoding allows for error detection, an important remaining question is whether these error detection properties can be extended into scenarios where we apply gates (i.e. simulations) on the data within the code.  We show in the following proposition that single errors within the syndrome measurements will never yield an undectable error. 

\begin{proposition}
With the choices in Definition~\ref{def:our gse}, a single-qubit error during syndrome measurement cannot produce an undetectable fault.
\end{proposition}

\begin{proof}
We first consider syndrome measurement for the non-boundary loop operator $IYXZYXZI$. A syndrome measurement circuit is shown in Figure~\ref{fig:syndrome}. It uses Pauli-controlled-Pauli gates; the rule for propagating Pauli errors through such gates is shown in Figure~\ref{fig:commutation_rule}.

The error detection behavior of this measurement is as follows: A single-qubit error on any of the vertex qubits will trigger one of the following consequences:
\begin{enumerate}
    \item The error commutes with all remaining gates, producing a one-qubit  error on one of the gate vertices, which is detectable, though this syndrome measurement does not detect it.
    \item The error fails to commute with one Pauli-controlled-Pauli gate, propagating an $X$ error to the ancilla qubit, which commutes with all remaining gates and is detected by the $Z$ measurement.
\end{enumerate}

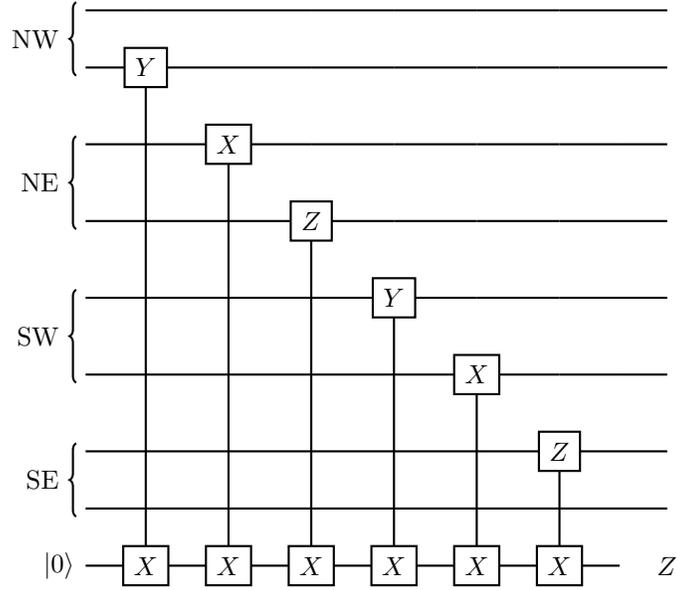
\begin{figure}
\centering
\begin{quantikz}
 \lstick[wires=2]{NW} & \qw & \qw & \qw & \qw & \qw & \qw & \qw & \qw \\
   &\gate{Y} \vqw{7} & \qw &\qw & \qw &\qw &\qw & \qw &\qw \\
 \lstick[wires=2]{NE} & \qw & \gate{X} \vqw{6}& \qw & \qw & \qw & \qw & \qw & \qw \\
 & \qw & \qw & \gate{Z} \vqw{5}& \qw & \qw & \qw & \qw & \qw\\
 \lstick[wires=2]{SW} & \qw & \qw & \qw & \gate{Y} \vqw{4} & \qw & \qw & \qw & \qw \\
 & \qw & \qw & \qw & \qw & \gate{X} \vqw{3}  & \qw & \qw & \qw \\
 \lstick[wires=2]{SE} & \qw & \qw & \qw & \qw & \qw & \gate{Z} \vqw{2} & \qw  & \qw\\
 \qw & \qw & \qw & \qw & \qw & \qw & \qw & \qw & \qw\\
 \lstick{$\ket{0}$} & \gate{X} & \gate{X} & \gate{X} & \gate{X} & \gate{X} & \gate{X} & \qw & Z\\
\end{quantikz}
\caption{Syndrome measurement for the weight six loop operator $IYXZYXZI$.}
\label{fig:syndrome}
\end{figure}

\begin{figure}
    \begin{subfigure}[b]{\linewidth/2}
        \centering
        \begin{quantikz}
          \qw & \gate{P} & \gate{Q} \vqw{1} & \qw \\
          \qw & \qw & \gate{R} & \qw \\
        \end{quantikz}
        \;\;\;=
        \begin{quantikz}
            \qw & \gate{Q} \vqw{1} & \gate{P} & \qw \\
            \qw & \gate{R} & \qw & \qw \\
        \end{quantikz}
        \caption{if $PQ = QP$}
    \end{subfigure}
    \begin{subfigure}[b]{\linewidth/2}
        \centering
        \begin{quantikz}
          \qw & \gate{P} & \gate{Q} \vqw{1} & \qw \\
          \qw & \qw & \gate{R} & \qw \\
        \end{quantikz}
        \;\;\;=
        \begin{quantikz}
            \qw & \gate{Q} \vqw{1} & \gate{P} & \qw \\
            \qw & \gate{R} & \gate{R} & \qw \\
        \end{quantikz}
        \caption{if $PQ = -QP$}
    \end{subfigure}

    \vspace{.5cm}
    \begin{subfigure}[b]{\linewidth/2}
        \centering
        \begin{quantikz}
          \qw & \gate{P} \vqw{1} & \qw & \qw \\
          \qw & \gate{Q} & \gate{S} \vqw{1} & \qw \\
          \qw & \qw & \gate{R} & \qw \\
        \end{quantikz}
        \;\;\;=
        \begin{quantikz}
          \qw & \qw & \gate{P} \vqw{1} & \qw \\
          \qw & \gate{S} \vqw{1} & \gate{Q} & \qw \\
          \qw & \gate{R} & \qw & \qw \\
        \end{quantikz}
        \caption{if $Q S = S Q$}
    \end{subfigure}
    \begin{subfigure}[b]{\linewidth/2}
        \centering
        \begin{quantikz}
          \qw & \gate{P} \vqw{1} & \qw & \qw \\
          \qw & \gate{Q} & \gate{S} \vqw{1} & \qw \\
          \qw & \qw & \gate{R} & \qw \\
        \end{quantikz}
        \;\;\;=
        \begin{quantikz}
          \qw & \qw & \gate{P} \vqw{1} & \gate{P} \vqw{2} & \qw \\
          \qw & \gate{S} \vqw{1} & \gate{Q} & \qw & \qw \\
          \qw & \gate{R} & \qw & \gate{R} & \qw \\
        \end{quantikz}
        \caption{if $Q S = - S Q$}
    \end{subfigure}
    \vspace{.5cm}
    \caption{Commutation rules for general Pauli-controlled Pauli operators. Here, $P$, $Q$, $R$, and $S$ are Pauli gates.}
    \label{fig:commutation_rule}
\end{figure}
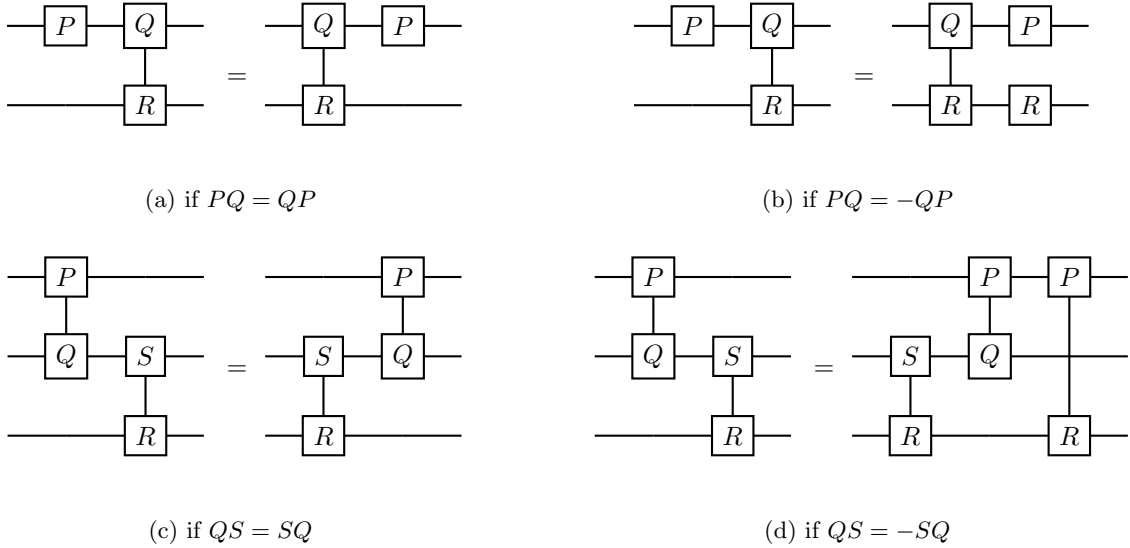

A single-qubit error on the ancilla propagates differently. Here, an $X$ error is detected by ancilla measurement, but a $Z$ error is not. A $Z$ error propagates to one vertex qubit error at each two-qubit gate which follows the error on the ancilla. In order to retain detectability, this propagated product of vertex-qubit errors must be detectable or trivial. For example, a $Z$-error occurring immediately after the ancilla-qubit is initialized propagates an $IYXZYXZI$ operator on the vertex qubits. This is acceptable since $IYXZYXZI$ is a stabilizer and does not affect the encoded state. An error proceeding, say, the second $Y$-controlled Pauli, on the other hand, propagates a nontrivial $IIIIYXZI$ vertex-qubit error which must be detected later. 

For a weight-six stabilizer there are five ancilla $Z$-error propagations to consider. All of them produce nontrivial syndrome measurements, as is shown in Figure~\ref{ancilla error syndromes}, values for which are computed using single-qubit syndromes shown in Figure~\ref{single qubit error syndromes}. Therefore none of these errors are logical operations or stabilizers.

It remains to consider the case of weight three loop operators on boundary edges. Here again, errors on the vertex qubits propagate to single-vertex-qubit errors, and errors on the ancilla qubit are suffixes of the measured loop operator. For a weight three operator, every such suffix is either a loop operator, a weight one operator, or an operator which differs from a weight one operator by a loop operator, and thus logically equivalent to a weight one error. The first of these is trivial, the other two are detectable.
\end{proof}

The lack of undetectable errors is to some extent the result of fortunate choices for the $\gamma_{i,j}$ and their edge mappings. However, if undetectable errors do occur, it may be possible to eliminate them by reordering the Pauli-controlled-Pauli gates, since these commute with each other, but not with all errors.

\begin{proposition}\label{prop:bj measurement} With the choices made in Definition~\ref{def:our gse}, a single-qubit error made during measurement of the $B_j$ operators cannot produce an undetectable fault.
\end{proposition}
Note that this statement is only meaningful when operations are performed after the $B_j$-operator measurements. Also, in most circumstances fermion number parity considerations allow detection of a single incorrect $B_j$ operator measurement. 

\begin{figure}
    \centering
    \begin{quantikz}
        \qw & \gate{Z}\vqw{2} & \qw & \qw \\
        \qw & \qw{} & \gate{Y}\vqw{1} & \qw \\
        \lstick{$\ket{0}$} & \gate{X} & \gate{X} & \qw & Z
    \end{quantikz}
    \caption{Measurement circuit for a $B_j$ operator}
    \label{fig:bj_measurement}
\end{figure}
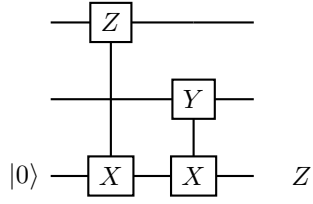
\begin{proof}[Proof of Proposition~\ref{prop:bj measurement}]
See Figure~\ref{fig:bj_measurement}. The only one-qubit error which can propagate to an error on two vertex qubits is a $Z$-error on the ancilla qubit prior to the first gate. Since at this point the ancilla is in state $\ket{0}$, $Z$ acts trivially.
\end{proof}

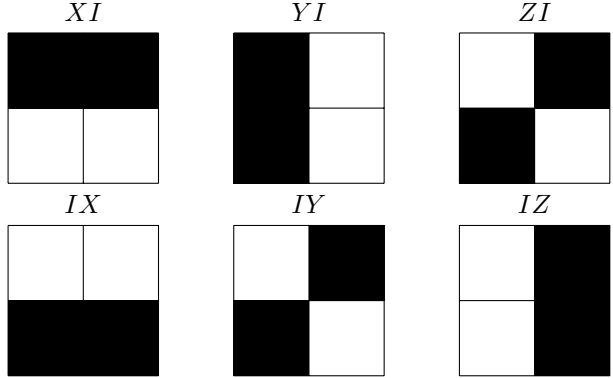
\begin{figure}
    \centering
    \begin{tikzpicture}
    \begin{scope}[local bounding box = XI]
      \fill (0,1) rectangle ++ (1,1);
      \fill (1,1) rectangle ++ (1,1);
      \draw (0,0) grid (2,2); 
    \end{scope}
    
    \begin{scope}[xshift=3cm, local bounding box = YI]
      \fill (0,0) rectangle ++ (1,1);
      \fill (0,1) rectangle ++ (1,1);
      \draw (0,0) grid (2,2); 
    \end{scope}
    
    \begin{scope}[xshift=6cm, local bounding box = ZI]
      \fill (0,0) rectangle ++ (1,1);
      \fill (1,1) rectangle ++ (1,1);
      \draw (0,0) grid (2,2); 
    \end{scope}
    
    \path[nodes={text depth=0.25ex, above, font=\sffamily}]
      (XI.north) node{$XI$}
      (YI.north) node{$YI$}
      (ZI.north) node{$ZI$}
      ;
    \end{tikzpicture}

    \begin{tikzpicture}
    \begin{scope}[local bounding box = IX]
      \fill (0,0) rectangle ++ (1,1);
      \fill (1,0) rectangle ++ (1,1);
      \draw (0,0) grid (2,2); 
    \end{scope}
    
    \begin{scope}[xshift=3cm, local bounding box = IY]
      \fill (1,1) rectangle ++ (1,1);
      \fill (0,0) rectangle ++ (1,1);
      \draw (0,0) grid (2,2); 
    \end{scope}
    
    \begin{scope}[xshift=6cm, local bounding box = IZ]
      \fill (1,0) rectangle ++ (1,1);
      \fill (1,1) rectangle ++ (1,1);
      \draw (0,0) grid (2,2); 
    \end{scope}
    
    \path[nodes={text depth=0.25ex, above, font=\sffamily}]
      (IX.north) node{$IX$}
      (IY.north) node{$IY$}
      (IZ.north) node{$IZ$}
      ;
    \end{tikzpicture}

    \caption{Single-qubit error syndromes on the vertex qubits as encoded in Figure~\ref{fig:spinless hubbard weight}}
    \label{single qubit error syndromes}
\end{figure}

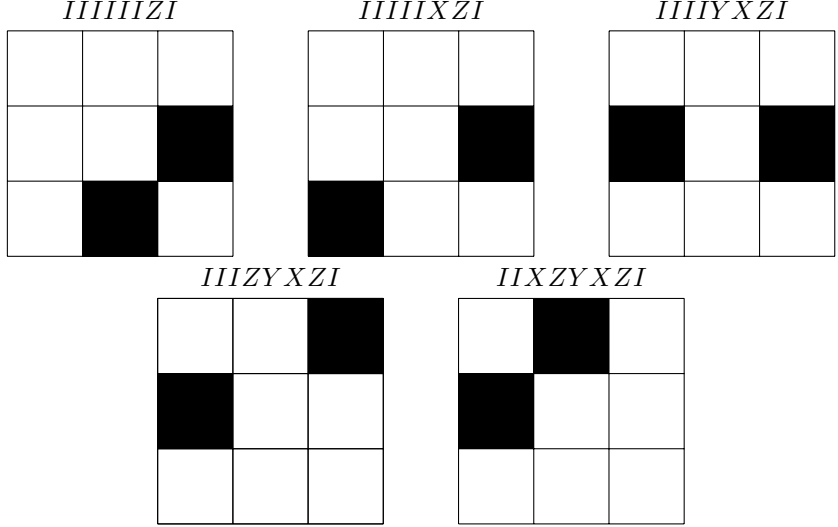
\begin{figure}
\centering
\begin{tikzpicture} 
\begin{scope}[local bounding box=ZI]
 \fill (2,1) rectangle ++ (1,1); 
 \fill (1,0) rectangle ++ (1,1); 
 \draw (0,0) grid (3,3); 
\end{scope}

\begin{scope}[xshift=4cm,local bounding box=XZI]
 \fill (0,0) rectangle ++ (1,1); 
 \fill (2,1) rectangle ++ (1,1); 
 \draw (0,0) grid (3,3); 
\end{scope}

\begin{scope}[xshift=8cm,local bounding box=YXZI]
 \fill (0,1) rectangle ++ (1,1); 
 \fill (2,1) rectangle ++ (1,1); 
 \draw (0,0) grid (3,3); 
\end{scope}

\path[nodes={text depth=0.25ex,above,font=\sffamily}] 
    (ZI.north) node{$IIIIIIZI$}
    (XZI.north) node{$IIIIIXZI$}
    (YXZI.north) node{$IIIIYXZI$}
    ;
\end{tikzpicture} 

\begin{tikzpicture} 
\begin{scope}[local bounding box=ZYXZI]
 \fill (0,1) rectangle ++ (1,1); 
 \fill (2,2) rectangle ++ (1,1); 
 \draw (0,0) grid (3,3); 
 \draw (0,0) grid (3,3); 
\end{scope}

\begin{scope}[xshift=4cm,local bounding box=XZYXZI]
 \fill (0,1) rectangle ++ (1,1); 
 \fill (1,2) rectangle ++ (1,1); 
 \draw (0,0) grid (3,3); 
\end{scope}

\path[nodes={text depth=0.25ex,above,font=\sffamily}] 
    (ZYXZI.north) node{$IIIZYXZI$}
    (XZYXZI.north) node{$IIXZYXZI$}
    ;
\end{tikzpicture} 

\caption{Ancilla error syndromes for the syndrome measurement in Figure~\ref{fig:syndrome}.}
\label{ancilla error syndromes}
\end{figure}

\subsection{Fault-detecting evolution}
\label{sec:evolution}
Fault-detection for evolution operators is more complicated. Since the edge and vertex operators are logical operators, evolutions of edge-and-vertex algebra elements are logical operators as well. Detectable errors which occur prior to a logical evolution produce the same syndrome measurements as if they had occurred after evolution. However, errors may also occur during evolution, and must be detected.

The most straightforward method to evolve a Pauli operator $P$ is to find a Clifford operator $U$ such that $U \circ P \circ U^\dagger = ZII\ldots$, then implement $e^{-i Pt}$ as $U^\dagger \circ (e^{-i Z t}I\ldots I) \circ U$. Conjugation by $U$ propagates the $ZI\ldots I$ evolution to a $P$ evolution. Such a circuit cannot avoid propagating a $ZII$ error that occurs during $ZII$ evolution to a logical operator.

This is an issue which is encountered, and solved, in the context of error-correction; one distills high-fidelity magic states \cite{PhysRevA.71.022316} and uses them to construct a circuit approximating $e^{-i Z t}$. However, in a NISQ context, the required resource overhead for magic states is undesirable. The cheapest solutions are to use a native hardware gate that evolves two qubits simultaneously (e.g. $e^{-i ZZ t}$), if one is available, or, just fail to detect this particular error. If undetectable one-qubit errors can be limited to a small number of cases it may make sense to mitigate them using probabilistic methods \cite{PhysRevLett.119.180509}.

\begin{figure}
    \centering
    \begin{subfigure}[b]{\textwidth}
        \centering
        \begin{tabular}{|l|l|}
            \multicolumn{2}{c}{Horizontal edge} \\
            \hline
            $A_{jk}$ & $IZXY$ \\
            $A_{jk}B_k $ & $IZYI$ \\
            $B_j A_{jk}$ & $ZXXY$ \\
            $B_j A_{jk} B_j$ & $ZXYI$ \\
            \hline
        \end{tabular}
        \hspace{1cm}
        \begin{tabular}{|l|l|}
            \multicolumn{2}{c}{Vertical edge} \\
            \hline
            $A_{jk}$ & $IXYY$ \\
            $A_{jk}B_k $ & $IXXI$ \\
            $B_j A_{jk}$ & $ZZYY$ \\
            $B_j A_{jk} B_j$ & $ZZXI$ \\
            \hline
        \end{tabular}
        \hspace{1cm}
        \begin{tabular}{|l|l|}
            \multicolumn{2}{c}{Vertices-only} \\
            \hline
            $B_j$ & $ZYII$ \\
            $B_k$ & $IIZY$ \\
            $B_j B_k$ & $ZYZY$ \\
            \hline
        \end{tabular}
    
        \caption{Every pair of vertices which bound an edge supports seven logical operators, the values of which depend on the edge orientation.}
        \label{two-vertex operators single}
    \end{subfigure} 
    
    \vspace{.5cm}
    
    \begin{subfigure}[b]{\textwidth}
        \centering
        \begin{tabular}{|l|l|}
            \multicolumn{2}{c}{Top edge} \\
            \hline
            $A'_{jk}$ & $YYYY$ \\
            $A'_{jk}B_k $ & $YYXI$ \\
            $B_j A'_{jk}$ & $XIYY$ \\
            $B_j A'_{jk} B_k$ & $XIXI$ \\
            $A_{jk}A'_{jk}$ & $YXZI$ \\
            $A_{jk}A'_{jk}B_k $ & $YXIY$ \\
            $B_j A_{jk} A'_{jk}$ & $XZZI$ \\
            $B_j A_{jk} A'_{jk} B_k$ & $XZIY$ \\
            \hline
        \end{tabular}
        \hspace{.1cm}
        \begin{tabular}{|l|l|}
            \multicolumn{2}{c}{Bottom edge} \\
            \hline
            $A'_{jk}$ & $IXIX$ \\
            $A'_{jk}B_k $ & $IXZZ$ \\
            $B_j A'_{jk}$ & $ZZIX$ \\
            $B_j A'_{jk} B_j$ & $ZZZZ$ \\
            $A_{jk} A'_{jk}$ & $IYXZ$ \\
            $A_{jk} A'_{jk}B_k $ & $IYYX$ \\
            $B_j A_{jk} A'_{jk}$ & $ZIXZ$ \\
            $B_j A_{jk} A'_{jk} B_j$ & $ZIYX$ \\
            \hline
        \end{tabular}
        \hspace{.1cm}
        \begin{tabular}{|l|l|}
            \multicolumn{2}{c}{Left edge} \\
            \hline
            $A'_{jk}$ & $XYXY$ \\
            $A'_{jk}B_k $ & $XYYI$ \\
            $B_j A'_{jk}$ & $YIXY$ \\
            $B_j A'_{jk} B_k$ & $YIYI$ \\
            $A_{jk} A'_{jk}$ & $XZZI$ \\
            $A_{jk} A'_{jk}B_k $ & $XZIY$ \\
            $B_j A_{jk} A'_{jk}$ & $YXZI$ \\
            $B_j A_{jk} A'_{jk} B_j$ & $YXIY$ \\
            \hline
        \end{tabular}
        \hspace{.1cm}
        \begin{tabular}{|l|l|}
            \multicolumn{2}{c}{Right edge} \\
            \hline
            $A'_{jk}$ & $IZIZ$ \\
            $A'_{jk}B_k $ & $IZZX$ \\
            $B_j A'_{jk}$ & $ZXIZ$ \\
            $B_j A'_{jk} B_k$ & $ZXZX$ \\
            $A_{jk} A'_{jk}$ & $IYYX$ \\
            $A_{jk} A'_{jk}B_k $ & $IYXZ$ \\
            $B_j A_{jk} A'_{jk}$ & $ZIYX$ \\
            $B_j A_{jk} A'_{jk} B_j$ & $ZIXZ$ \\
            \hline
        \end{tabular}
    
        \caption{If such a pair of vertices bound two pairs of edges (because it lies on the boundary), it admits eight additional logical operators, the values of which depend on which part of the boundary contains the edge.}
        \label{fig:two-vertex operators double}
    \end{subfigure}
    \caption{Two-vertex operators, with scalars omitted, using the GSE assignments shown in Figure~\ref{fig:spinless hubbard weight}. In these tables, the conventional qubit order is given by placing the left- or top-most qubit first. Signs are omitted for simplicity and depend on edge orientation.}
    \label{fig:two-vertex operators}
\end{figure}

\begin{proposition}
With the choices in Definition~\ref{def:our gse}, if sufficient two-qubit native hardware gates are available, evolution of local fermionic algebra operators can be performed so that one-qubit errors during evolution $e^{-i P t}$ do not propagate to undetectable errors. If only $CNOT$ and arbitrary one-qubit gates are available, a propagated undetectable one-qubit error takes the form of the operator $p$ being evolved, possibly along with time-reversal of the evolution.
\end{proposition}

\begin{proof}

The choices in Definition~\ref{def:our gse} give the two-vertex logical operators shown in Figure~\ref{fig:two-vertex operators}. For each given pair of interacting vertices, the edge orientation and whether the edge lies on the boundary together determine the logical operators supported on those vertices. As the figure enumerates, among the $4^4 = 256$ possible two-vertex Pauli operators, considered up to scalar, a pair of vertices which bound a single edge admit seven logical operators (plus the trivial operator); a pair bounding a double edge admits fifteen. For example, an $XIXI$ error is a logical error on a (doubled) top edge, but nowhere else.

Figure~\ref{fig:four-qubit eval plain} shows an evolution circuit for a four-qubit operator. The figure assumes that the fourth qubit, the bottom one, is acted upon by a nontrivial Pauli gate. If this is not the case, a smaller circuit can be constructed similarly evolving the second or third qubit instead, using a controlled $Q_2$ or $Q_3$ gate on the second or third strand respectively instead of controlled $Q_4$ on the fourth strand. In each case $P_i$ and $Q_i$ anticommute, but $Q_i$ is otherwise arbitrary. The possibilities for nontrivial propagations of single-qubit errors are summarized in Figure~\ref{fig:four qubit evolution errors}. There, an error is a Pauli gate $Q$ which anticommutes with some gate $P_i$ in the evolved Pauli operator.

\begin{figure}
    \centering
    \begin{subfigure}[b]{\linewidth}
    \centering
    \begin{quantikz}
      \qw &  \gate{P_1} \vqw{3} & \qw & \qw & \qw & \qw &  \qw & \gate{P_1} \vqw{3}  & \qw \\
      \qw & \qw & \gate{P_2}\vqw{2} & \qw & \qw & \qw & \gate{P_2} \vqw{2} & \qw & \qw \\
      \qw & \qw & \qw & \gate{P_3} \vqw{1}& \qw & \gate{P_3} \vqw{1}& \qw & \qw & \qw \\
      \qw & \gate{Q_4} &\gate{Q_4} &\gate{Q_4} & \gate{e^{-i P_4 t}} & \gate{Q_4} & \gate{Q_4} & \gate{Q_4} & \qw \\
    \end{quantikz}
    \caption{Without full error-detection.}
    \label{fig:four-qubit eval plain}
    \end{subfigure}
    \begin{subfigure}[b]{\linewidth}
    \centering
    \begin{quantikz}
      \qw &  \gate{P_1} \vqw{3} & \qw & \qw & \qw & \qw &  \qw & \gate{P_1} \vqw{3} & \qw\\
      \qw & \qw & \gate{P_2}\vqw{2} & \qw & \qw & \qw & \gate{P_2} \vqw{2} & \qw & \qw \\
      \qw & \qw & \qw & \qw & \gate[wires=2]{e^{-i P_3 P_4 t}}  & \qw & \qw & \qw & \qw \\
      \qw & \gate{Q_4} &\gate{Q_4} & \qw & \qw & \qw & \gate{Q_4} & \gate{Q_4} & \qw \\
    \end{quantikz}
    \caption{With a two-qubit native gate to prevent one-qubit errors adjacent to the evolution from propagating to logical errors.}
    \label{fig:four-qubit eval fixed}
    \end{subfigure}
    \caption{Quantum circuits for evolution of $e^{-iPt}$, where $P = P_1 P_2 P_3 P_4$. $Q_4$ is any Pauli gate not commuting with $P_4$.}
    \label{fig:four-qubit eval}
\end{figure}
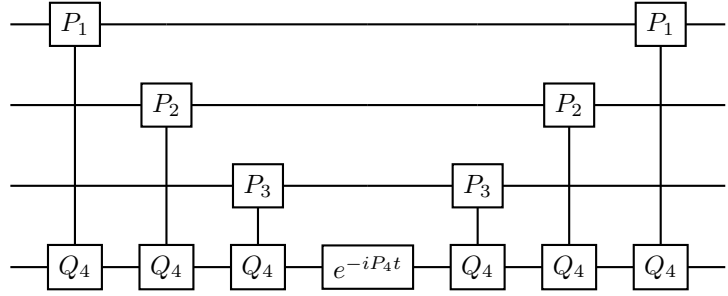
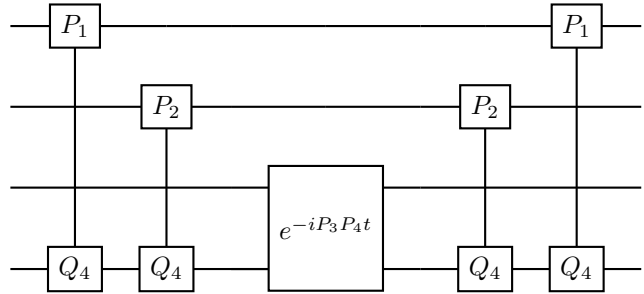

\begin{figure}
    \centering
    \begin{tabular}{|c|c|c|c|}
    \hline
    Error strand & Largest strand & Comparison & Propagated result \\
    \hline

    \multirow{5}{*}{1} & 2-4 & $P_1=Q$ & $QIII$ \\
    \cline{2-4}
    & 4 & \multirow{3}{*}{$P_1 \neq Q$}& $QIIQ_4$ \\
    & 3 & & $QIQ_3I$ \\
    & 2 & & $QQ_2II$ \\
    \hline

    \multirow{5}{*}{2} & 3-4 & $P_2=Q$ & $IQII$ \\
    \cline{2-4}
    & 4 & \multirow{2}{*}{$P_2\neq Q$} & $IQIQ_4$ \\
    & 3 & & $IQQ_3I$ \\
    \cline{2-4}
    & \multirow{2}{*}{2} & $Q_2 = Q$ & $IQII$ \\
    \cline{3-4}
    & & $Q_2 \neq Q$ & $P_1QII$ \\
    \hline
    
    \multirow{5}{*}{3} & \multirow{2}{*}{4} & $P_3=Q$ & $IIQI$ \\
    & & $P_3 \neq Q$ & $IIQQ_4$ \\
    \cline{2-4}
    & \multirow{3}{*}{3} & $Q_3 = Q$ & $IIQI$ \\
    \cline{3-4}
    & & \multirow{2}{*}{$Q_3 \neq Q$} & $P_1 I QI$ or \\
    & & & $P_1 P_2 QI$. \\
    \hline
    
    \multirow{4}{*}{4} & \multirow{4}{*}{4} & $Q_4 = Q$ & $IIIQ$ \\
    \cline{3-4}
    & & \multirow{3}{*}{$Q_4 \neq Q$} & $P_1IIQ$ or \\
    & & & $P_1P_2IQ$ or \\
    & & & $P_1P_2P_3Q$. \\
    \hline
\end{tabular}

    \caption{Propagation possibilities for a Pauli gate error of type $Q$ during evolution of a weight two or greater four-qubit operator of the type shown in Figure~\ref{fig:four-qubit eval plain}. Errors which occur prior to or after evolution, possibly after application of far commutativity relations, are not listed, since evolution of a logical operator is a logical operator and does not change syndrome measurements.}
    \label{fig:four qubit evolution errors}
\end{figure}

We will show that the following modifications to \ref{fig:four-qubit eval plain} allow us to construct evolution gates which do not propagate undetectable errors, aside from the Pauli operator that is being evolved:
\begin{enumerate}
    \item Make one of two choices of Pauli gate for the controlled Pauli operators,
    \item Add an ancilla and flag one of the qubits,
    \item Reverse the order of the qubits (equivalent to reflecting Figure~\ref{fig:four-qubit eval} about a horizontal axis).
\end{enumerate}

For ease of discussion, We call operators with reversed numbering {\em reflected operators} and say that the circuits that evaluate them produce reflected errors when evaluated. For example, the operator $ZXZX$ corresponds to reflected operator $XZXZ$, and the reflected $IQIQ_4$-type error $IYIY$ that can occur when $XZXZ$ is evolved is a $YIYI$ error in the unreflected code space. To organize the discussion, we consider it an $IQIQ_4$ error on a reflected operator, rather than a $Q_4IQI$ error.

We proceed by considering the error-types shown in Figure~\ref{fig:four qubit evolution errors}. Some can be easily seen to be detectable, with no ancilla, whether they occur in an unreflected or reflected evolution.
\begin{enumerate}
\item Errors of type $QIIQ_4$ or $P_1IIQ$, as no logical operator or its reflection is of this form,
\item One-qubit propagated errors, as well as those which differ from the evolved logical operator or evolved reflected operator by a single qubit, namely $P_1QII$, $P_1P_2QI$, or $P_1P_2P_3Q$,
\item Errors of type $QQ_2II$. These errors can only occur when evolving operators with support on the first two qubits. The only such operators are $ZYII$ and reflected operator $YZII$, but $Q_2 \ne Y$ (and for the reflection, $Q_2 \ne Z$) by construction so $QQ_2II$ is not a logical operator.
\end{enumerate}
This leaves errors of types $QIQ_3I$, $IQQ_3I$,$P_1IQI$, $IQIQ_4$, $IIQQ_4$,  and $P_1P_2IQ$ to be dealt with, as well as the errors $P_1P_2II$, $P_1P_2P_3II$ and $P_1P_2P_3P_4$; these last three are logical errors corresponding to evolved operators and reflected operators, which the proposition does not require us to correct. Errors of the first three types ($QIQ_3I$, $IQQ_3I$, and $P_1IQI$) occur exclusively on operators which act nontrivially on the third qubit, and trivially on the fourth. Errors of the next three types ($IQIQ_4$, $IIQQ_4$,  and $P_1P_2IQ$) occur on operators which act nontrivially on the fourth qubit.

We first account for the $QIQ_3I$ and $IQIQ_4$ errors. Weight-two operators of the forms $AIBI$ and $IAIB$ only occur on the boundary edges and there is exactly one for each edge type. Thus a $QIQ_3I$ or $IQIQ_4$ error is only a logical error when it occurs on a boundary edge and is logical for that boundary edge type. Such operators are never evolved as the construction gives them weight zero, but they are still logical operators, equivalent to the edge operators for which they are doubled edges, which we must detect as errors.

$QIQ_3I$ errors occur during $P_1P_2P_3I$ gates on top and left boundary edges, and in reflection on the bottom and right boundary edges, and there is only one logical error $AIBI$ supported on each boundary edge. If $P_3 = B$, neither choice of $Q_3 \ne P_3$ makes a $QIQ_3I$ error a logical operator, otherwise there is still a choice of $Q_3$ such that $QIQ_3I$ is not a logical operator. In similar fashion, for each operator $P_1P_2P_3P_4$, there is at least one choice of $Q_4$ such that an $IQIQ_4$ error, or reflected error, is not a logical operator.

\begin{figure}
    \centering
    \begin{quantikz}
      \qw & \gate{P_1} \vqw{2} & \qw & \qw & \qw & \qw & \qw &  \qw & \qw & \gate{P_1} \vqw{2} & \qw\\
      \qw & \qw & \gate[style=yellow]{P_2}\vqw{3} & \gate{P_2}\vqw{1} & \qw & \qw & \qw & \gate{P_2} \vqw{1} & \gate[style=yellow]{P_2} \vqw{3} & \qw & \qw \\
      \qw & \gate{Q_3} & \qw & \gate{Q_3} & \qw & \gate{e^{-i P_3 t}} & \qw & \gate{Q_3} & \qw & \gate{Q_3} & \qw \\
      \qw & \qw & \qw & \qw & \qw & \qw & \qw & \qw & \qw & \qw & \qw \\
      \lstick{$\ket{0}$} & \gate[style=yellow]{H} & \ctrl{} & \qw & \qw & \qw & \qw & \qw & \ctrl{} & \gate[style=yellow]{H} & \qw & Z
    \end{quantikz}

    \caption{A $P_1 P_2 P_3 I$ operator with a flag on the second qubit to detect $IQQ_4I$ errors. Here the $Z$-measurement on the last qubit discretizes and detects $Q$ errors on the second qubit which could otherwise propagate to undetected $IQQ_3I$ errors.}
    \label{fig:flag second qubit}
\end{figure}
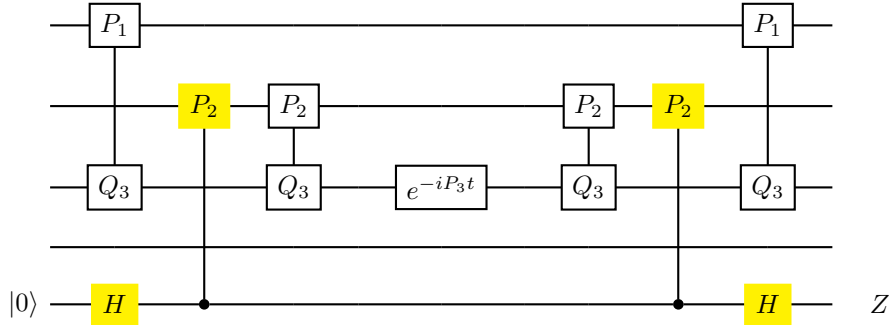

We next consider $IQQ_3I$ errors, which occur only in $P_1P_2P_3I$ gates. Exactly one operator of type $IABI$ occurs for each edge type, $IZYI$ for horizontal edges (boundary or not), with reflected operator $IYZI$, and $IXXI$ for vertical edges. If the choice of $Q_3$ was not previously forced, we can avoid $IQQ_3I$ logical errors by choosing it now. Among non-reflected operators, the only instance in which logical $IQQ_3I$ errors can't be avoided in this way is the top edge operator $YXZI$, for which we chose $Q_3=Y$ already to avoid $XIXI$ logical errors. Propagated $IZYI$ errors can be detected using a flag qubit in the manner of \cite{chao-reichardt-few} to detect errors on the second qubit, as shown in Figure~\ref{fig:flag second qubit}. For vertical edges the story is similar; for $XZZI$ we require $Q_3 = X$ in order to avoid logical $YIYI$ errors, and can flag the second qubit to detect the $IQQ_3I$ error. Reflected operators can be treated similarly; the bottom edge reflected operator $XZYI$ requires $Q_3=Z$ to avoid $XIXI$ reflected errors, and right edge reflected operator $XYYI$ forces $Q_3 = X$. The $IYZI$ and $IXXI$ reflected operator errors can both be detected with a flag on the second qubit.

Next we consider $P_1IQI$ errors. These are only possible when two logical operators with support on the first three qubits have the same first qubit. This occurs in two unreflected cases: $XZZI$ can produce an $XIXI$ error, and $YXZI$ can produce a $YIYI$ error. Unfortunately it is not possible to flag the third qubit as the flag operations don't commute with the evolution operator. Thus we evolve these operators in the reflected code space (also reflecting the flag qubits we added previously). The reflected $P_1P_2P_3I$ operator $IP'_2P'_3P'_4$ has the wrong qubit support to produce $P_1IQI$ logical errors.

Next, the unique $IIQQ_4$ unreflected error $IIZY$ and reflected error $IIYZ$ can be avoided by choosing $Q_4$ to avoid the unique $IIQQ_4$ logical error $IIZY$, except when $Q_4$ has already been chosen unfortunately, i.e as $Q_4 = Y$ in the non-reflected case, $Q_4 = Z$ in the reflected case. If a good choice of $Q_4$ is unavailable, we can flag the third qubit to detect the $Q$ error. Previously, we added ancilla qubits to operators where the fourth qubit was trivial, which can't produce $IIQQ_4$ errors, and to $YXZI$ and $XZZI$, which we are evolving in reflection. Neither of these can produce an $IIZY$ error. Thus we have added at most one ancilla to each evolution at this point.

Errors of the form $P_1P_2IQ$ are only logical errors, according to Figure~\ref{fig:two-vertex operators}, on the boundary and must propagate an evolved an operator $P_1P_2P_3P_4$, $P_4 \ne I$, which differs from $P_1P_2IQ$ by a $B_k$ operator. However, for no such operator is it the case that both $P_1P_2IQ$ and $Q'IP_3P_4$ are logical operators. If it were otherwise, since $Q'IP_3P_4$ differs from $P_1P_2P_3P_4$ by $B_j$, $Q'IIQ$ would appear in Figure~\ref{fig:two-vertex operators}. Therefore, evolution of at most one of $P_1P_2P_3P_4$ or its reflection can produce logical $P_1P_2IQ$ errors. Any operator which has not already been reflected may be reflected to avoid $P_1 P_2IQ$ errors.

It remains to show that we have not reflected a $P_1P_2P_3I$ operator which can produce a $P_1IQI$ logical error to produce a reflected operator $P'_1P'_2P'_3P'_4$ which can produce a $P'_1P'_2IQ'$ logical error. In this case, $P'_1 = I$ and then $P_1 I QI$ and $IP'_2IQ$ are reflections of each other, since each boundary edge only supports a single logical operator among both of these shapes. But then $Q = P'_2$, which is impossible since $P_3 = P'_2$ and $Q \ne P_3$ by construction.

Thus every logical operator $P_1P_2P_3P_4$ may be evolved or reflection-evolved, at least one of the two, possibly with an ancilla qubit, so that single-qubit errors do not propagate to undetectable logical errors, except for the case when $P_1P_2P_3P_4$ is propagated as an error.

\begin{figure}
    \centering
    \begin{tabular}{|c|c|c|c|}
    \hline
    Error strand & Largest strand & Comparison & Propagated result \\
    \hline
    
    1 & 2 & $P_1 \neq Q$ & $QIII$ \\
    \hline

    \multirow{2}{*}{2} & 3 & $P_2\ne Q$ & $IQII$ \\
    \cline{2-4}
    & 2 & $Q_2 \neq Q$ & $IQII$ \\
    \hline
    
    \multirow{2}{*}{3} & 4 & $P_3 \neq Q$ & $IIQI$ \\
    \cline{2-4}
    & 3 & $Q_3 \neq Q$ & $P_1 I QI$ \\
    \hline
    
    \multirow{2}{*}{4} & \multirow{2}{*}{4} & \multirow{2}{*}{$Q_4 \neq Q$} & $P_1IIQ$ or \\
    & & & $P_1P_2IQ$ \\
    \hline
\end{tabular}

    \caption{Propagation possibilities for a Pauli gate error of type $Q$ during evolution of a weight two or greater four-qubit operator of the type shown in Figure~\ref{fig:four-qubit eval fixed}. Only those types of error propagations which differ from those in Figure~\ref{fig:four qubit evolution errors} are shown.}
    \label{fig:four qubit evolution errors fixed}
\end{figure}

The evolved operator $P_1P_2P_3P_4$ is propagated as an error when a single-qubit $Q = P_i$ error occurs just prior to or just after evolution $e^{-iP_it}$ on the evolved strand. Since the evolution operator is designed to convert a $P_i$ operator into a logical operator, the standard Pauli operator evolution circuit cannot detect this error. In addition to these $P_i$ errors, it is not possible to detect an error in $t$. $P_1P_2P_3P_4$ errors can be eliminated, in our error model, if a native hardware gate $e^{-iP_{i-1}P_it}$ is available, giving the circuit shown in Figure~\ref{fig:four-qubit eval fixed}. This gate propagates one-qubit errors to one-qubit errors, possibly replacing $t$ with $-t$ in the evolution in the process. The circuit in Figure~\ref{fig:four-qubit eval fixed} has some different error propagations than those listed in Figure~\ref{fig:four qubit evolution errors}; these are summarized in Figure~\ref{fig:four qubit evolution errors fixed}. In particular, the undetectable error types are removed, and no new error types are introduced. Thus, with two-qubit native hardware Pauli evolutions, one-qubit errors are detectable.

\end{proof}

\section{Error-detection under reduced connectivity}\label{sec:reduced connectivity}
As written, the above circuits can be implemented by creating one ancilla for each plaquette of the (torus-embedded) Hamiltonian interaction graph.  The assumed connectivity is shown in Figure~\ref{fig:full connectivity}. In particular, each ancilla has connections to each of the eight qubits assigned to the four corners of the plaquette. The two qubits within each vertex are connected, and for each edge, one of the four qubits in the edge is connected with all three others. This produces ancillas with degree eight connectivity and vertex qubits with degree seven or ten connectivity. These connectivity assumptions are generous and will not be satisfied by some NISQ hardware.

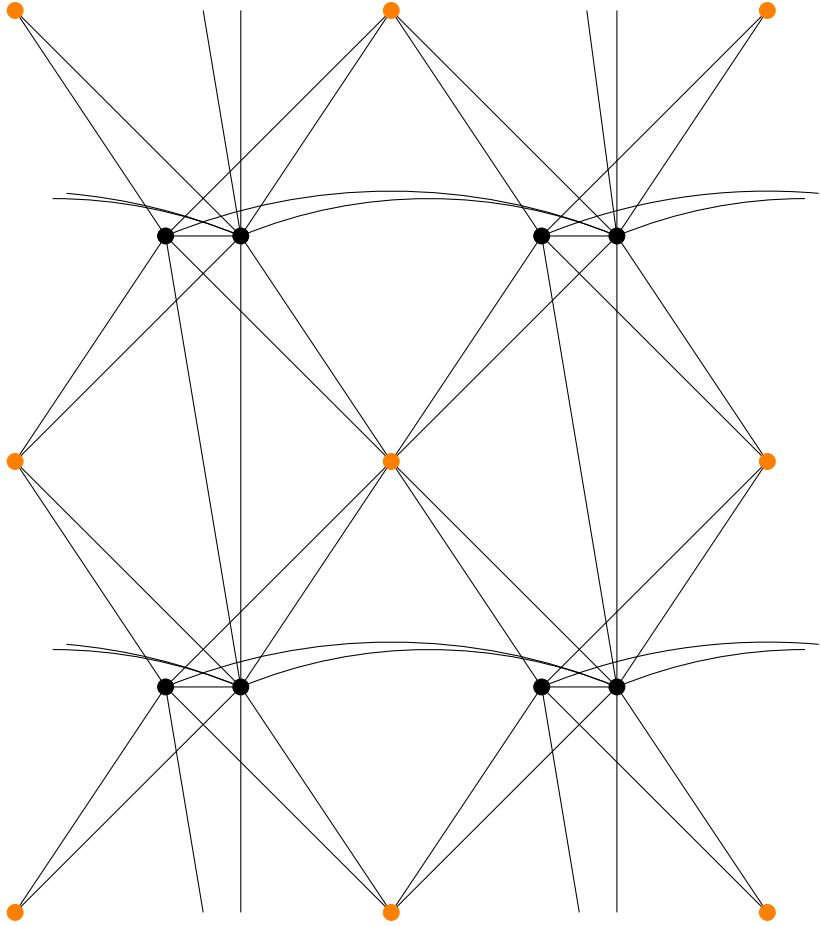
\begin{figure}
    \centering
    \begin{tikzpicture}
    \draw (0,0) -- (1,0);
    \draw (0,6) -- (1,6);
    \draw (5,0) -- (6,0);
    \draw (5,6) -- (6,6);
    \draw (0,0) -- (3,3);
    \draw (1,0) -- (3,3);
    \draw (0,6) -- (3,3);
    \draw (1,6) -- (3,3);
    \draw (5,0) -- (3,3);
    \draw (6,0) -- (3,3);
    \draw (5,6) -- (3,3);
    \draw (6,6) -- (3,3);
    \draw (5,6) -- (3,9);
    \draw (5,6) -- (8,9);
    \draw (5,6) -- (8,3);
    \draw (6,6) -- (3,9);
    \draw (6,6) -- (8,9);
    \draw (6,6) -- (8,3);
    \draw (0,6) -- (-2,9);
    \draw (0,6) -- (-2,3);
    \draw (0,6) -- (3,9);
    \draw (1,6) -- (-2,9);
    \draw (1,6) -- (-2,3);
    \draw (1,6) -- (3,9);
    \draw (5,0) -- (3,-3);
    \draw (5,0) -- (8,-3);
    \draw (5,0) -- (8,3);
    \draw (6,0) -- (3,-3);
    \draw (6,0) -- (8,-3);
    \draw (6,0) -- (8,3);
    \draw (0,0) -- (3,-3);
    \draw (0,0) -- (-2,-3);
    \draw (0,0) -- (-2,3);
    \draw (1,0) -- (3,-3);
    \draw (1,0) -- (-2,-3);
    \draw (1,0) -- (-2,3);

    \draw (0,0)  arc (112.5:67.5:{6/(2*sin(22.5))} );
    \draw (1,0)  arc (112.5:67.5:{5/(2*sin(22.5))} );
    \draw (0,6)  arc (112.5:67.5:{6/(2*sin(22.5))} );
    \draw (1,6)  arc (112.5:67.5:{5/(2*sin(22.5))} );
    \draw (5,0)  arc (112.5:85:{6/(2*sin(22.5))} );
    \draw (6,0)  arc (112.5:90:{5/(2*sin(22.5))} );
    \draw (5,6)  arc (112.5:85:{6/(2*sin(22.5))} );
    \draw (6,6)  arc (112.5:90:{5/(2*sin(22.5))} );
    \draw (1,0)  arc (67.5:85:{6/(2*sin(22.5))} );
    \draw (1,0)  arc (67.5:90:{5/(2*sin(22.5))} );
    \draw (1,6)  arc (67.5:85:{6/(2*sin(22.5))} );
    \draw (1,6)  arc (67.5:90:{5/(2*sin(22.5))} );
    \draw (1,0) -- (0,6);
    \draw (1,0) -- (1,6);
    \draw (6,0) -- (5,6);
    \draw (6,0) -- (6,6);
    \draw (1,6) -- (.5,9);
    \draw (1,6) -- (1,9);
    \draw (6,6) -- (5.6,9);
    \draw (6,6) -- (6,9);
    \draw (1,-3) -- (1,0);
    \draw (.5,-3) -- (0,0);
    \draw (6,-3) -- (6,0);
    \draw (5.5,-3) -- (5,0);
    
    \filldraw[color=orange] (-2,-3) circle (3pt);
    \filldraw[color=orange] (3,-3) circle (3pt);
    \filldraw[color=orange] (8,-3) circle (3pt);
    \filldraw (0,0) circle (3pt);
    \filldraw (1,0) circle (3pt);
    \filldraw (5,0) circle (3pt);
    \filldraw (6,0) circle (3pt);
    \filldraw[color=orange] (-2,3) circle (3pt);
    \filldraw[color=orange] (3,3) circle (3pt);
    \filldraw[color=orange] (8,3) circle (3pt);
    \filldraw (0,6) circle (3pt);
    \filldraw (1,6) circle (3pt);
    \filldraw (5,6) circle (3pt);
    \filldraw (6,6) circle (3pt);
    \filldraw[color=orange] (-2,9) circle (3pt);
    \filldraw[color=orange] (3,9) circle (3pt);
    \filldraw[color=orange] (8,9) circle (3pt);
    \end{tikzpicture}
    \caption{Full qubit connectivity requirements for the eight vertex qubits of a plaquette, and the nine ancilla qubits which interact with them. Ancillas are orange.}
    \label{fig:full connectivity}
\end{figure}

 \begin{proposition}
 The connectivity requirements for performing fault-tolerant syndrome measurement (Figure~\ref{fig:syndrome}), evolution (e.g. Figure~\ref{fig:four-qubit eval}, Figure~\ref{fig:flag second qubit}), and measurement (Figure~\ref{fig:bj_measurement}) circuits, as shown in Figure~\ref{fig:full connectivity}, can be reduced to the requirements shown in Figure~\ref{fig:less connectivity} while still detecting arbitrary one-qubit errors, at the cost of some additional swaps which produce only detectable errors.
 \end{proposition}
 
 \begin{proof}
 
  The standard swap gate circuit (composed of three non-commuting controlled-not operations) propagates a one-qubit $Q$ error to a $QQ$ errors. The code has been chosen so that no $QQ$ error is a logical error. To detect arbitrary one-qubit errors it suffices to show that $QQ$ errors produced by swaps performed during an evolution or syndrome measurement does not propagate to logical errors.

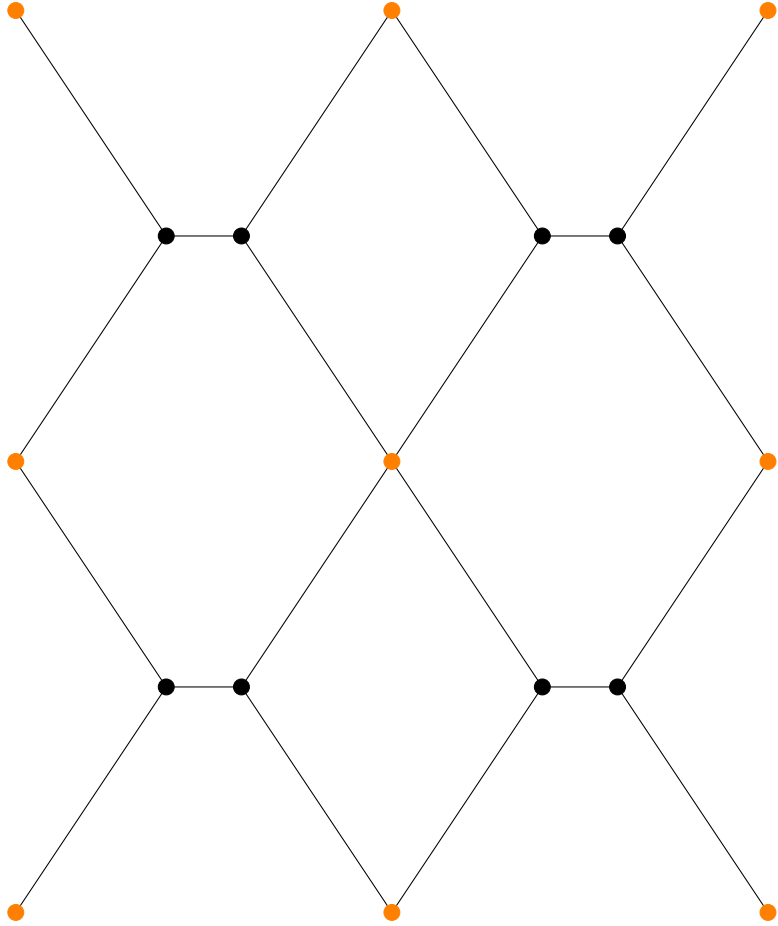
\begin{figure}
    \centering
    \begin{tikzpicture}
    \draw (0,0) -- (1,0);
    \draw (0,6) -- (1,6);
    \draw (5,0) -- (6,0);
    \draw (5,6) -- (6,6);
    \draw (1,0) -- (3,3);
    \draw (1,6) -- (3,3);
    \draw (5,0) -- (3,3);
    \draw (5,6) -- (3,3);
    \draw (6,6) -- (8,9);
    \draw (6,6) -- (8,3);
    \draw (5,6) -- (3,9);
    \draw (1,6) -- (3,9);
    \draw (0,6) -- (-2,9);
    \draw (0,6) -- (-2,3);
    \draw (6,0) -- (8,-3);
    \draw (6,0) -- (8,3);
    \draw (5,0) -- (3,-3);
    \draw (1,0) -- (3,-3);
    \draw (0,0) -- (-2,-3);
    \draw (0,0) -- (-2,3);

    \filldraw[color=orange] (-2,-3) circle (3pt);
    \filldraw[color=orange] (3,-3) circle (3pt);
    \filldraw[color=orange] (8,-3) circle (3pt);
    \filldraw (0,0) circle (3pt);
    \filldraw (1,0) circle (3pt);
    \filldraw (5,0) circle (3pt);
    \filldraw (6,0) circle (3pt);
    \filldraw[color=orange] (-2,3) circle (3pt);
    \filldraw[color=orange] (3,3) circle (3pt);
    \filldraw[color=orange] (8,3) circle (3pt);
    \filldraw (0,6) circle (3pt);
    \filldraw (1,6) circle (3pt);
    \filldraw (5,6) circle (3pt);
    \filldraw (6,6) circle (3pt);
    \filldraw[color=orange] (-2,9) circle (3pt);
    \filldraw[color=orange] (3,9) circle (3pt);
    \filldraw[color=orange] (8,9) circle (3pt);
    \end{tikzpicture}
    \caption{Reduced qubit connectivity requirements for syndrome measurement and evolution. Ancillas are orange.}
    \label{fig:less connectivity}
\end{figure}

For syndrome measurements, a reduced connectivity version of Figure~\ref{fig:syndrome} is shown in Figure~\ref{fig:syndrome reduced connectivity}. The only swaps needed are between qubits on the same vertex. Each $QQ$ error produced by such swaps propagates to a $QQ$ error, along with an $X$ error on the flag qubit, at the time of syndrome measurement. Such errors are detectable.

Evolution circuits such as those in Figure~\ref{fig:four-qubit eval} and Figure~\ref{fig:flag second qubit}, under reduced connectivity, have connectivity restrictions which depend on the orientation of evolved sites within the lattice. Figure~\ref{fig:evolution swaps} shows the qubit connectivities for various operators. The desired permutations can be accomplished with swaps of two forms:
\begin{enumerate}
    \item Swap an ancilla qubit with an adjacent vertex qubit,
    \item Swap the two qubits of a single vertex to or from their original positions.
\end{enumerate}

A reduced-connectivity version of the evolution circuit of Figure~\ref{fig:four-qubit eval fixed}, for vertical edges, is shown in Figure~\ref{fig:four-qubit eval local}. A vertex-vertex swap error propagates to either a $QQ$ error on a paired-vertex qubit pair, or a $QQIQ_4$ error. A paired-vertex $QQ$ error is never a logical error by the code's construction. Figure~\ref{fig:two-vertex operators} shows that $QQIQ_4$ can only be a logical error if it is a reflected $YYIX$ error on a top edge or an unreflected $ZZIX$ error on a bottom edge. The last (respectively first) gate must then anticommute with $Q=Y$ (respectively $Q=Z$), but no such gates occur for top (respectively bottom) edges in the table. Thus, even in the presence of swaps, errors are detectable.

For horizontal edges the initial qubit ordering will differ slightly. The circuit in Figure~\ref{fig:four-qubit eval local} may be modified to a horizontal edge circuit by conjugating with a $(1,2)$ swap. The conjugating swaps can propagate a vertex-vertex swap error, which propagates to a detectable $QQ$ error. 
\end{proof}

To reduce the degree of the interaction graph to three, each ancilla may be replaced with a pair of qubits connected by an edge, to form a hexagonal lattice. This increases the qubit cost of representing a fermion site from roughly three qubits per site to four.

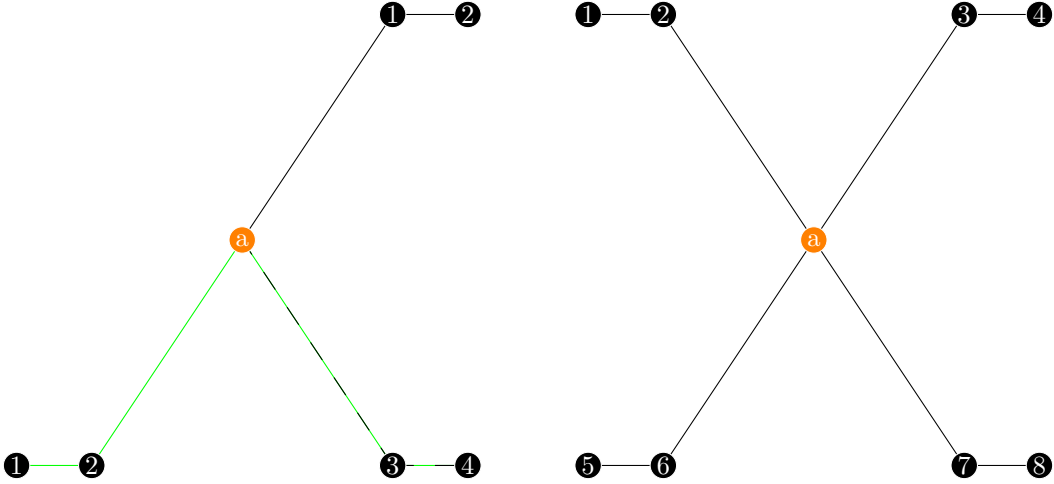
\begin{figure}
    \centering
    \begin{tikzpicture}
        \draw[green] (0,0) -- (1,0);
        \draw[dash pattern=on 8pt off 8pt,black] (0,0) -- (1,0);
        \draw[green] (0,0) -- (-2,3);
        \draw[dash pattern=on 8pt off 8pt,black] (0,0) -- (-2,3);
        \draw[green] (-2,3) -- (-4,0);
        \draw[green] (-4,0) -- (-5,0);
        \draw (-2,3) -- (0,6);
        \draw (0,6) -- (1,6);
        
        \draw[color=white] (0,6) node[circle,minimum size=10pt,inner sep=0pt,fill=black,draw]{1};
        \draw[color=white] (1,6) node[circle,minimum size=10pt,inner sep=0pt,fill=black,draw]{2};
        \draw[color=white] (0,0) node[circle,minimum size=10pt,inner sep=0pt,fill=black,draw]{3};
        \draw[color=white] (1,0) node[circle,minimum size=10pt,inner sep=0pt,fill=black,draw]{4};
        \draw[color=white] (-5,0) node[circle,minimum size=10pt,inner sep=0pt,fill=black,draw]{1};
        \draw[color=white] (-4,0) node[circle,minimum size=10pt,inner sep=0pt,fill=black,draw]{2};
        \draw[color=white] (-2,3) node[circle,minimum size=10pt,inner sep=0pt,fill=orange,draw]{a};
    \end{tikzpicture}
    \hspace{1cm}
    \begin{tikzpicture}
        \draw (0,0) -- (1,0);
        \draw (0,0) -- (-2,3);
        \draw (-2,3) -- (-4,0);
        \draw (-4,0) -- (-5,0);
        \draw (-2,3) -- (0,6);
        \draw (0,6) -- (1,6);
        \draw (-2,3) -- (-4,6);
        \draw (-4,6) -- (-5,6);
        
        \draw[color=white] (0,6) node[circle,minimum size=10pt,inner sep=0pt,fill=black,draw]{3};
        \draw[color=white] (1,6) node[circle,minimum size=10pt,inner sep=0pt,fill=black,draw]{4};
        \draw[color=white] (0,0) node[circle,minimum size=10pt,inner sep=0pt,fill=black,draw]{7};
        \draw[color=white] (1,0) node[circle,minimum size=10pt,inner sep=0pt,fill=black,draw]{8};
        \draw[color=white] (-5,0) node[circle,minimum size=10pt,inner sep=0pt,fill=black,draw]{5};
        \draw[color=white] (-4,0) node[circle,minimum size=10pt,inner sep=0pt,fill=black,draw]{6};
        \draw[color=white] (-5,6) node[circle,minimum size=10pt,inner sep=0pt,fill=black,draw]{1};
        \draw[color=white] (-4,6) node[circle,minimum size=10pt,inner sep=0pt,fill=black,draw]{2};
        \draw[color=white] (-2,3) node[circle,minimum size=10pt,inner sep=0pt,fill=orange,draw]{a};
    \end{tikzpicture}
    \caption[]{Qubits used for evolving local fermionic algebra operators and syndrome measurements in the reduced connectivity. Ancillas are orange. The qubit connectivity for various operators is as follows:
    \begin{enumerate*}
    \item horizontal $A_{jk}$: green -- left image,
    \item vertical $A_{jk}$: black --- left image,
    \item $B_j$: a horizontal edge --- right image,
    \item $B_j$ measurement: each of the ancilla's four arms --- either image,
    \item syndrome measurement: all edges --- right image.
    \end{enumerate*}
    Labelled qubit indices in the left image illustrate our conventional written order of qubits in horizontal and vertical $A_{jk}$ operators. The right image shows the the qubit order for loop operators.
    }
    \label{fig:evolution swaps}
\end{figure}

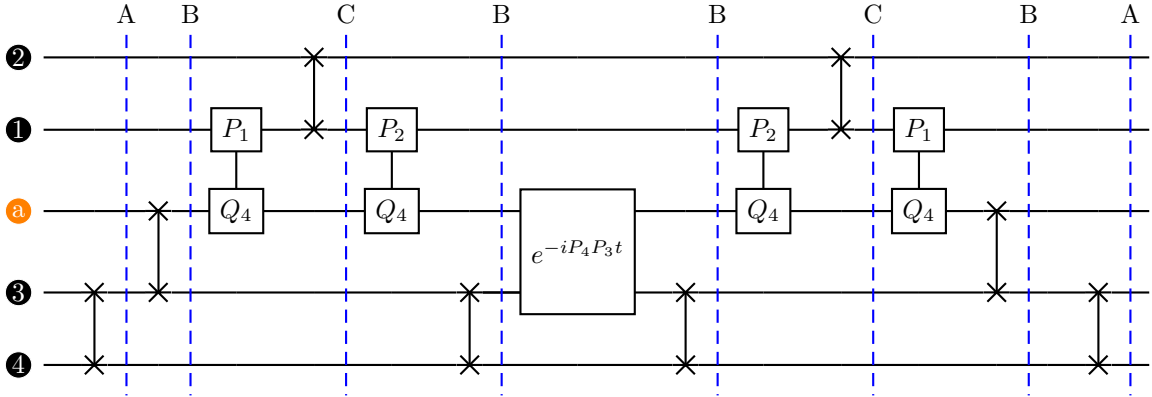
\begin{figure}
    \centering
    \begin{quantikz}
      \lstick{\begin{tikzpicture}\draw[color=white] node[circle,minimum size=10pt,inner sep=0pt,fill=black,draw]{2};\end{tikzpicture}} & \qw & \qw & \qw & \swap{1} \slice[style=blue]{C} & \qw & \qw & \qw & \qw & \qw & \swap{1} \slice[style=blue]{C} & \qw & \qw & \qw & \qw & \qw \\
      \lstick{\begin{tikzpicture}\draw[color=white] node[circle,minimum size=10pt,inner sep=0pt,fill=black,draw]{1};\end{tikzpicture}} & \qw & \qw & \gate{P_1}\vqw{1} & \targX{} & \gate{P_2} \vqw{1} & \qw & \qw & \qw & \gate{P_2}\vqw{1} & \targX{} & \gate{P_1} \vqw{1} & \qw & \qw & \qw & \qw \\
      \lstick{\begin{tikzpicture}\draw[color=white] node[circle,minimum size=10pt,inner sep=0pt,fill=orange,draw]{a};\end{tikzpicture}} & \qw & \swap{1} \slice[style=blue]{B} & \gate{Q_4} & \qw & \gate{Q_4} & \qw & \gate[wires=2]{e^{-iP_4P_3t}} & \qw & \gate{Q_4} & \qw & \gate{Q_4} & \swap{1} \slice[style=blue]{B} & \qw & \qw & \qw \\
      \lstick{\begin{tikzpicture}\draw[color=white] node[circle,minimum size=10pt,inner sep=0pt,fill=black,draw]{3};\end{tikzpicture}} & \swap{1} \slice[style=blue]{A} &  \targX{} & \qw & \qw & \qw & \swap{1} \slice[style=blue]{B} & \qw & \swap{1} \slice[style=blue]{B} & \qw & \qw & \qw & \targX{} & \qw & \swap{1} \slice[style=blue]{A} & \qw \\
      \lstick{\begin{tikzpicture}\draw[color=white] node[circle,minimum size=10pt,inner sep=0pt,fill=black,draw]{4};\end{tikzpicture}} & \targX{} & \qw & \qw & \qw & \qw & \targX{} & \qw & \targX{} & \qw & \qw & \qw & \qw & \qw & \targX{} & \qw 
    \end{quantikz}
    \caption{Quantum circuits for evolution of $e^{-iPt}$, where $P = P_1 P_2 P_3 P_4$, with swaps, under reduced connectivity. The circuit shown is for a vertical edge operator. For a horizontal edge operator, qubits $1$ and $2$ are swapped in the connectivity graph; we assume this is accomplished with an additional pair of swaps on qubits $1$ and $2$ at the beginning and end of the circuit. Single-qubit errors during swaps propagate according to the slice labels at the top of the figure: to an $IIQQ$ error at $A$, a single-qubit vertex error plus an ancilla error at $B$, and to a $QQII$ or $QQIQ_4$ error at $C$.}
    \label{fig:four-qubit eval local}
\end{figure}
\begin{figure}
    \centering
    \begin{quantikz}
    \lstick{\begin{tikzpicture}\draw[color=white] node[circle,minimum size=10pt,inner sep=0pt,fill=black,draw]{1};\end{tikzpicture}} & \qw & \qw & \qw & \qw & \qw & \qw & \qw & \qw & \qw \\
    \lstick{\begin{tikzpicture}\draw[color=white] node[circle,minimum size=10pt,inner sep=0pt,fill=black,draw]{2};\end{tikzpicture}} & \gate{Y} \vqw{3}  & \qw & \qw & \qw & \qw & \qw & \qw & \qw & \qw \\
    \lstick{\begin{tikzpicture}\draw[color=white] node[circle,minimum size=10pt,inner sep=0pt,fill=black,draw]{3};\end{tikzpicture}} & \qw & \gate{X}\vqw{2} & \swap{1} & \gate{Z}\vqw{2} & \swap{1} & \qw & \qw & \qw & \qw \\
    \lstick{\begin{tikzpicture}\draw[color=white] node[circle,minimum size=10pt,inner sep=0pt,fill=black,draw]{4};\end{tikzpicture}} & \qw & \qw & \targX{} & \qw & \targX{} & \qw & \qw & \qw & \qw \\
    %
    \lstick{\begin{tikzpicture}\draw[color=white] node[circle,minimum size=10pt,inner sep=0pt,fill=orange,draw]{a};\end{tikzpicture}$\ket{0}$} & \gate{X} & \gate{X} & \qw & \gate{X} & \gate{X}\vqw{2} & \qw & \gate{X}\vqw{2} & \gate{X}\vqw{3} & \qw & Z\\
    \lstick{\begin{tikzpicture}\draw[color=white] node[circle,minimum size=10pt,inner sep=0pt,fill=black,draw]{5};\end{tikzpicture}} & \qw & \qw & \qw & \swap{1} & \qw & \swap{1} & \qw & \qw & \qw \\
    \lstick{\begin{tikzpicture}\draw[color=white] node[circle,minimum size=10pt,inner sep=0pt,fill=black,draw]{6};\end{tikzpicture}} & \qw & \qw & \qw & \targX{} & \gate{Y} & \targX{} & \gate{X} & \qw & \qw \\
    \lstick{\begin{tikzpicture}\draw[color=white] node[circle,minimum size=10pt,inner sep=0pt,fill=black,draw]{7};\end{tikzpicture}} & \qw & \qw & \qw & \qw & \qw & \qw & \qw & \gate{Z} & \qw \\
    \lstick{\begin{tikzpicture}\draw[color=white] node[circle,minimum size=10pt,inner sep=0pt,fill=black,draw]{8};\end{tikzpicture}} & \qw & \qw & \qw & \qw & \qw & \qw & \qw & \qw & \qw \\
    \end{quantikz}
    \caption{Evaluation circuit for an $IYXZYXZI$ syndrome measurement under reduced connectivity. The physical connectivity constraints are as shown in Figure~\ref{fig:evolution swaps}.}
    \label{fig:syndrome reduced connectivity}
\end{figure}
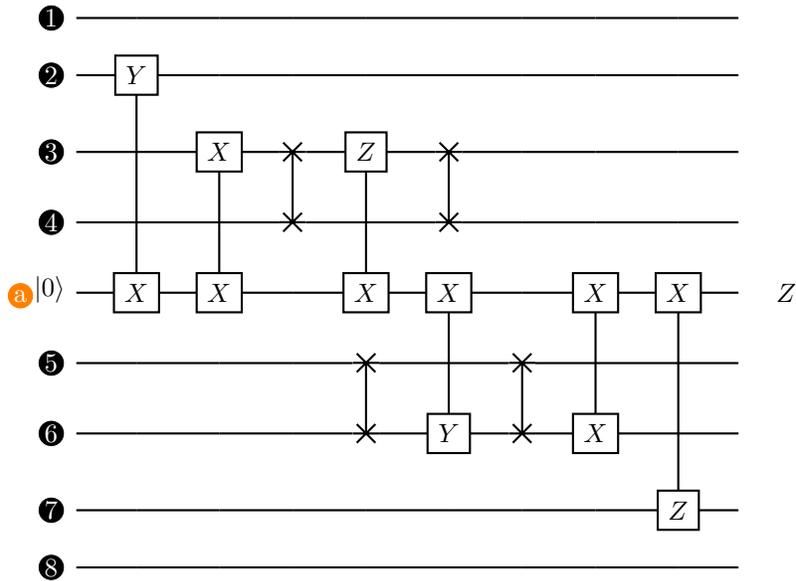

\section{Resource and performance bounds for a VQE circuit}\label{sec:performance}

In order to establish that there is a performance regime in which the GSE's error detection properties are useful, in this section we derive parameterized circuits for a Variational Quantum Eigensolver~\cite{peruzzo2014variational,grimsley2019adaptive,cerezo2021variational}. Our ansatz states will be constructed from product states using parameterized Hamiltonian Variational Ansatz (HVA) circuits. We choose the HVA method because it has an efficient implementation in the edge and vertex algebra, namely, we prepare an initial state and then for parameters $U_j, U_{jk}, \hat{U}_{jk}$, we evolve
\[\prod e^{-i U_{jk} i B_{(j,\uparrow)} A_{(j,\uparrow),(k,\uparrow)}} \prod e^{-i U_{jk} i A_{(j,\uparrow),(k,\uparrow)} B_{(k,\uparrow)}} \prod e^{i U_j B_{(j,\uparrow)}} \prod e^{-i \hat U_{jk} B_{(j,\uparrow)} B_{(k,\uparrow)}},\]
with the products taken over site orbital indices $j$ and positively oriented edges $jk$.


Circuit depth and two-qubit gate costs are computed as follows. To prepare the initial definite-occupancy state, syndrome measurements for each square plaquette and each bigon plaquette on the boundary must be performed. The circuit in Figure~\ref{fig:syndrome reduced connectivity} performs a syndrome measurement by interacting a plaquette's vertex qubits with the ancilla lying in its center, using a northwest, northeast, southwest, southeast sweep. For each plaquette edge, the plaquette vertex-qubit interactions occur either both before or both after the vertex-interactions for the plaquette sharing an edge boundary. Thus, syndrome measurements commute and may be performed simultaneously.

Figure~\ref{fig:syndrome reduced connectivity} shows that a loop operator syndrome measurement requires at most ten two-qubit gates and a $Z$ measurement. In total, preparation of a definite-occupancy state requires at most $(N-1)(M-1) * 10$ two-qubit gates, circuit depth at most eight, and $(N-1)(M-1)$ $Z$-measurements.

The evolution operators do not all commute, however the $e^{-i U_jB_{(j,\uparrow)}}$ commute and have disjoint support. These may all be performed in a single layer.

The operators $e^{-i \hat U_{jk} B_{(j,\uparrow)}B_{(k,\uparrow)}}$ commute; the horizontal edge operators may be evolved simultaneously by an argument similar to that made above for syndrome measurements, as can the vertical edge operators. Two layers of evolutions are required in total.

A pair of operators, each of the form $e^{-i U_{jk} B_{(j,\uparrow)}A_{(j,\uparrow),(k,\uparrow)}}$ or $e^{-i U_{jk} A_{(j,\uparrow),(k,\uparrow)}B_{(k,\uparrow)}}$, will commute in two cases: when their edge operators share an even number of vertices, and when exactly one of their vertex operators lies on the vertex shared by both edges. All operators of these two forms may be evolved in four layers, by evolving the operator types $e^{-i U_{jk} B_{(j,\uparrow)}A_{(j,\uparrow),(k,\uparrow)}}$ and $e^{-i U_{jk} A_{(j,\uparrow),(k,\uparrow)}B_{(k,\uparrow)}}$ types separately, each type's stage containing a layer for horizontal and a layer for vertical edges. Note that doubled edges have zero weight and are not evolved.

This results in a total of $1 + 2 + 4 = 7$ evolution layers to prepare the ansatz state. We assume as before that the two-qubit $B_i$ evolutions require one hardware-native two-qubit gate. For the other two types of evolutions, the circuit shown in Figure~\ref{fig:evolution swaps} shows the worst depth and gate-count cases: 13 two-qubit gates in 11 layers. Thus, the total number of layers for ansatz preparation is
\[1 + 11*(2 + 4) = 67.\]

The interaction graph has $mn$ vertices and $m(n-1) + n(m-1)$ edges, $m + n$ of which are doubled with zero interaction weight on doubled edges. Thus the total number of two-qubit gates for ansatz preparation is
\[m n + 13 \cdot 3(m(n-1) + n(m-1)) = 79 m n - 39(m+n-1).\]

Figure~\ref{fig:cost table} shows the two-qubit-gate and depth costs for $4\times4$, $8 \times 8$, and $16 \times 16$ planar lattices. Here, VQE is assumed to consist of zero-state preparation, ansatz evolution, application of a single Hamiltonian term (which costs zero two-qubit gates), time-reversed ansatz evolution, and measurement. The operators to be measured are the $B_j$ operators, but in the non-error-detected case, this may be accomplished by performing Pauli-gate basis measurements on the individual qubits and classically post-processing, without using two-qubit gates.

Error-detection is performed by repeating syndrome measurements performed during zero-state preparation. In this case, measuring the $B_j$ operators using single-qubit measurements plus classical post-processing is undesirable. Such measurements must be performed at the end of a computation and a fault during error-detection could then propagate to the measurements. We therefore measure the $B_j$ operators before error-detection, making the cost of error-detection equal to the cost of syndrome measurement, plus $4 M N$ two-qubit gates, assuming the $B_j$ measurement circuit shown in Figure~\ref{fig:B_j operator reduced}.

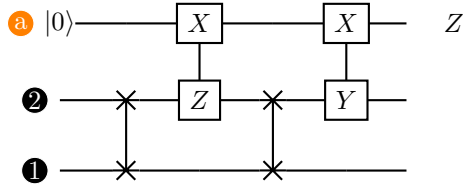
\begin{figure}
\centering
\begin{quantikz}
  \lstick{\begin{tikzpicture}\draw[color=white] node[circle,minimum size=10pt,inner sep=0pt,fill=orange,draw]{a};\end{tikzpicture}} \ket{0} & \qw & \gate{X} \vqw{1} & \qw & \gate{X} \vqw{1} & \qw & Z \\
  \lstick{\begin{tikzpicture}\draw[color=white] node[circle,minimum size=10pt,inner sep=0pt,fill=black,draw]{2};\end{tikzpicture}} & \swap{1} &  \gate{Z} & \swap{1} & \gate{Y} & \qw &\\
  \lstick{\begin{tikzpicture}\draw[color=white] node[circle,minimum size=10pt,inner sep=0pt,fill=black,draw]{1};\end{tikzpicture}} & \targX{} &  \qw & \targX{} & \qw & \qw &\\
\end{quantikz}
\caption{Measurement of a $B_j$ operator performed under the reduced connectivity assumptions shown in Figure~\ref{fig:less connectivity}.}
\label{fig:B_j operator reduced}
\end{figure}

\begin{figure}
\centering
\begin{tabular}{c|c|c|c|c|c|c|c|c|c}
    \multirow{2}{*}{m} & \multirow{2}{*}{n} & \multicolumn{2}{c|}{zero-state} & \multicolumn{2}{c|}{Ansatz} & \multicolumn{2}{c|}{VQE Total} & \multicolumn{2}{c}{Error Detected} \\
    & & Gates & Depth & Gates & Depth & Gates & Depth & Gates & Depth \\\hline
    4 & 4 & 170 & 10 & 1144 & 67 & 2458 & 144 & 2692 & 158 \\
8 & 8 & 650 & 10 & 5200 & 67 & 11050 & 144 & 11956 & 158 \\
16 & 16 & 2570 & 10 & 22048 & 67 & 46666 & 144 & 50260 & 158 \\
\end{tabular}

\caption{Two-qubit-gate and depth costs for planar lattices. An unmitigated VQE circuit consists of zero-state preparation, Ansatz state evolution, Hamiltonian evaluation (requiring zero two-qubit gates), conjugate Ansatz state evolution, and one-qubit measurements with classical post-processing. With error detection, the an error-detected circuit consists of zero-state preparation, Ansatz state evolution, Hamiltonian evaluation (requiring zero two-qubit gates), conjugate Ansatz state evolution, $B_j$ operator measurements, and error detection (using the same circuitry as zero-state preparation).}
\label{fig:cost table}
\end{figure}

Some error-detection metrics require sampling over the implemented circuit. The fraction of cases in which multiple errors prior to error-detection result in a detected error is one such metric; the rate of errors occurring during error-detection which result in a detection is another. For simplicity we shall, to begin with, assume that error-detection never occurs when multiple errors occur in a circuit, and that errors occurring during error detection are always undetected. These assumptions are extremely conservative and will be revised later.

\begin{lemma}\label{lmm:naive calculation}
Suppose a quantum circuit contains $c$ noisy two-qubit gates, producing an error on each of the qubits with (independent) probability $1-s$. After these gates, $d$ noisy two-qubit gates are employed to provide single-qubit error detection. Suppose further that error-detection never succeeds when multiple errors occur and that any error occurring during error detection is undetected. Then error detection increases the fraction of no-detected-error computations which are correct whenever
\begin{equation}\label{eq:threshold}
-2c s^{2c+2d} + s^{2c+2d+1} + s^{2d} - 1 > 0
\end{equation}
\end{lemma}
\begin{proof}
The probability of an unmitigated computation proceeding without error is
\begin{equation}\label{eq:p_g}
p_g = s^{2c}.
\end{equation}
The probability of an error occurring during error detection, after an error-free computation is
\[p_e = p_g (1-s^{2d}) = s^{2c} (1-s^{2d})\]
The probability of a single error occurring during computation, followed by an error-free error-detection, is
\begin{equation}\label{eq:pd}
p_d = 2c (1-s) s^{2c-1} s^{2d} = 2c (1-s) s^{2c + 2d - 1}.
\end{equation}
The probability of an error-detected computation proceeding with out error when no error is detected is
\begin{equation}\label{eq:p_g^ed}
p_g^{ed} = \frac{p_g - p_e}{1-p_d}
\end{equation}
The fraction of correct error-unmitigated computations is smaller than the fraction of correct error-mitigated computations when
\[p_g < \frac{p_g - p_e}{1-p_d},\]
which simplifies to
\[p_g p_d > p_e.\]
Substituting for $p_g$ and $p_e$ gives
\[s^{2c} p_d > s^{2c} (1-s^{2d})\]
which simplifies to
\begin{equation}\label{eq:pd agnostic improvement}
p_d > 1 - s^{2d}.
\end{equation}
Substituting for $p_d$ gives
\[2c (1-s) s^{2c +2d - 1} > (1-s^{2d}) \Leftrightarrow\]
\[\Leftrightarrow -2c s^{2c+2d} + 2c s^{2c+2d-1} + s^{2d} - 1 > 0\]
\end{proof}

\begin{figure}
\centering
\begin{tabular}{c|c|c|c|c|c|c}
    & & \multicolumn{2}{c|}{Improvement threshold} & \multicolumn{3}{c}{$s=.99999$} \\
    m & n & $s$ (\ref{eq:threshold}) & $p_g$ (\ref{eq:p_g})& $p_g$ (\ref{eq:p_g})& $p_d$ (\ref{eq:pd})& $p_g^{ed}$ (\ref{eq:p_g^ed})\\
\hline

    4 & 4 & 0.999544 & 0.106224 & 0.952029 & 0.046584 & 0.993882 \\
8 & 8 & 0.999891 & 0.089902 & 0.801716 & 0.173999 & 0.953170 \\
16 & 16 & 0.999974 & 0.088331 & 0.393244 & 0.341570 & 0.555822 \\
\end{tabular}

\caption{Threshold and high accuracy success probabilities. Improvement threshold is the error-free qubit gate rate $s$ at which error-detection provides a benefit which justifies its cost. $p_g$ is the fraction, in the absence of error detection, of circuit executions that complete without error. $p_d$ is the probability that a single-qubit error will occur and be detected. $p_g^{ed}$ is the fraction of circuits, after error detection and discarding, that complete without error.}
\label{fig:success probability}
\end{figure}

Success probability threshold $s$ values for improvement for different lattice sizes are shown in Figure~\ref{fig:success probability}. They show that error-detection in BKSF circuits can provide a benefit, compared to BKSF circuits without error detection, in error regimes above $1$ in $10^{-5}$ gates. However, the benefit is modest. On the other hand, our assumptions severely under count the number of detected errors.

Computing the detected error rate exactly would require computing the distribution of error syndromes from the distribution of errors. Instead, we estimate the probability $p_a$ of detecting an arbitrary error. If $p_a$ is known, the probability of an arbitrary error during computation, followed by an error detection, is
\[p_d = (1-s^{2c}) p_a.\]
Here it is not assumed that the error detection circuitry operates correctly, merely that in the presence of an error a nontrivial syndrome occurs.

From Equation~\ref{eq:pd agnostic improvement} in the proof of Lemma~\ref{lmm:naive calculation}, error detection improves the ratio of correct no-detected error calculations when
\begin{equation}\label{eq:arbitrary error improvement}
p_d > (1-s^{2d}) \Leftrightarrow (1-s^{2c})p_a > (1-s^{2d}) \Leftrightarrow -p_a s^{2c} + s^{2d} +p_a - 1 > 0.
\end{equation}

Though $p_a$ is difficult to compute, circuitry can be optimized in order to improve it. We propose that it should not be difficult to produce a scheme in which $p_a$ is greater than the probability that two single-qubit syndromes, sampled uniformly, produce a detectable error. This value can be computed as an estimate for $p_a$ and used to estimate error bounds.

Consideration of Figure~\ref{fig:two-vertex operators} shows that there is a single weight-two logical operator for each vertex, edge, and doubled edge, which may be obtained by ordering two single-qubit Pauli errors in two ways, plus there are six ways to obtain a trivial error on each of the vertices. Assuming a square lattice ($n=m$):
\begin{equation}\label{eq:p_a}
1-p_a \cong \frac{2*(m^2 + 2*m*(m-1) + m+m) + 6*m^2}{(m^2*6)^2}  = \frac{12*m^2}{(m^2*6)^2} = \frac{1}{3 m^2}.
\end{equation}

\begin{figure}
\centering
\begin{tabular}{c|c|c|c|c|c|c|c}
    \multirow{2}{*}{m} & \multirow{2}{*}{n} & \multirow{2}{*}{$p_a$ (\ref{eq:p_a})} & \multicolumn{2}{c|}{Improvement threshold} & \multicolumn{2}{c}{$p_g^{ed} = .95$} \\
    & & & $s$ (\ref{eq:arbitrary error improvement}) & $p_g$ (\ref{eq:p_g}) & $s$ (\ref{eq:success for performance}) & $p_g$ (\ref{eq:p_g}) \\\hline
    4 & 4 & $\frac{47}{48}$ & 0.991762 & 0.106224 & 0.999760 & 0.274634 \\
8 & 8 & $\frac{191}{192}$ & 0.997103 & 0.089902 & 0.999899 & 0.089346 \\
16 & 16 & $\frac{767}{768}$ & 0.999076 & 0.088331 & 0.999963 & 0.024251 \\
\end{tabular}
\caption{Circuit success probabilities and thresholds for a more optimistic error estimate. $p_a$ is the estimated probability of error-detection circuitry flagging an arbitrary error. Improvement threshold is the per-qubit-gate-success rate $s$ at which error detection improves accuracy. $p_g$ is the probability of the error-detected circuit completing without flagging an error.}
\label{fig:optimistic error}
\end{figure}

Figure~\ref{fig:optimistic error} shows $p_a$ values for various lattice sizes, and the odds of completing a circuit successfully at those thresholds. At threshold, circuits complete successfully so rarely that the output is essentially noise.

Suppose we want no-error-detected computations to be error-free with probability at least $p_g^{ed}$. We will assume that when ever computation proceeds correctly but error-detection does not, a correct result is flagged erroneously and discarded. Note that this assumption underestimates the fraction of correctly executed circuits and overestimates the rate at which circuits are discarded. Under this assumption, the undiscarded results consist of those with correct computation and error detection (which occur with probability $p_g = s^{2(c+d)}$) along with those with incorrect computation and unsuccessful error-detection (which are disjoint from the first set, and occur with probability $(1-s^c)(1-p_a)$). Then we need the following:
\begin{equation}\label{eq:success for performance}
\frac{s^{2c+2d}}{s^{2c + 2d} + (1-s^{2c})(1-p_a)} >  p_g^{ed} \Leftrightarrow \\
s^{2(c+d)}(1-p_g^{ed}) + s^{2c} p_g^{ed} (1-p_a) - p_g^{ed} > 0.
\end{equation}

The required $s$ accuracy thresholds to obtain runs which are $95\%$ accurate are shown in the rightmost columns of Figure~\ref{fig:optimistic error}.

\subsection{Repeated rounds of error detection}
Error-correction relies on repeated measurement of syndromes to keep accumulating errors within the correctable error space. It is reasonable to ask, in what regime are repeated rounds of error-detection helpful. Suppose a circuit contains two successive rounds of syndrome measurements, one after $c_1$ gates of the computation and one at the completion, after $c_2$ more gates, plus $d$ gates of error detection. For simplicity, and compatibility with the single-round notations, let us assume $c_1 = c_2 = c$. Then, using the formula from \ref{eq:success for performance}, the probability, after rejections due to error detection, of producing a correct calculation is estimated by
\[\left(\frac{s^{2c+2d}}{s^{2c + 2d} + (1-s^{2c})(1-p_a)}\right)^2,\]
and the probability of producing a correct calculation during a single round of error detection is
\[\frac{s^{4c+2d}}{s^{4c + 2d} + (1-s^{4c})(1-p_a)}.\]
Thus improvement occurs when
\[\left(\frac{s^{2c+2d}}{s^{2c + 2d} + (1-s^{2c})(1-p_a)}\right)^2 > \frac{s^{4c+2d}}{s^{4c + 2d} + (1-s^{4c})(1-p_a)} \Leftrightarrow\]
(by taking reciprocals and cancelling denominators),
\[\left(s^{2c + 2d} + (1-s^{2c})(1-p_a)\right)^2 < s^{4c+4d} + s^{2d}(1-s^{4c})(1-p_a) \Leftrightarrow\]
(by expanding terms and cancelling common summands),
\[2 s^{2c + 2d} (1-s^{2c})(1-p_a) + (1-s^{2c})^2(1-p_a)^2 < s^{2d} (1+s^{2c})(1-s^{2c})(1-p_a).\]
(by cancelling common factors),
\[2 s^{2c + 2d} + (1-s^{2c})(1-p_a) < s^{2d} (1+s^{2c}) \Leftrightarrow\]
(by solving for $s^{2d}$),
\[s^{2d} < \frac{-(1-s^{2c})(1-p_a)}{s^{2c} - 1} \Leftrightarrow\]
\[s^{2d} < 1-p_a.\]

This estimate gives a budget of
\begin{equation}\label{eq:d_b}
    d_b < \frac{\ln(1-p_a)}{2 \ln s}
\end{equation}
gates for error detection. 

\begin{figure}
    \centering
\begin{tabular}{c|c|c|c|c|c|c}
    $m$ & $n$ & $p_a (\ref{eq:p_a})$ & $d$ & $s$ & $d_b$ (\ref{eq:d_b}) \\\hline
    4&4& $\frac{47}{48}$ & 234 & 0.99 & 192 \\
&&  &  & 0.999 & 1934 \\
&&  &  & 0.9999 & 19355 \\
&&  &  & 0.99999 & 193559 \\
8&8& $\frac{191}{192}$ & 906 & 0.99 & 261 \\
&&  &  & 0.999 & 2627 \\
&&  &  & 0.9999 & 26286 \\
&&  &  & 0.99999 & 262873 \\
16&16& $\frac{767}{768}$ & 3594 & 0.99 & 330 \\
&&  &  & 0.999 & 3320 \\
&&  &  & 0.9999 & 33217 \\
&&  &  & 0.99999 & 332187 \\
\end{tabular}

    \caption{Gate count budgets $d_b$ for a second round of error detection for various lattice sizes $m \times n$ and gate accuracies $s$.}
    \label{fig:second round budget}
\end{figure}

Error-detection budgets for various lattice sizes and gate accuracies are shown in Figure~\ref{fig:second round budget}. They show that a second round of error detection can provide improved probability of success, at the cost of modest increase in the chance of detectable failure. Note that the improvement in accuracy will also be modest, as the estimate $p_a$ also serves as an estimate on the upper bound probability that an error which is detectable midway through a computation will remain so at final error-detection. If more rounds of error-detection are applied to the same computation, the probability of not discarding a computation decreases exponentially, while the benefits become increasingly modest; the unique contribution of each round of error detection is to catch errors occurring after the previous round and becoming undetectable between the current and next rounds. 

\section{Conclusion}
In this work we demonstrate that the Bravyi-Kitaev Superfast encoding can be used to perform error detection.  This is significant because it shows that advanced forms of error mitigation may be possible when using such an encoding that would otherwise not be possible without encoding the result inside an error detecting code.  We find that the process requires low weight stabilizer measurements and provide new approaches for performing a variational quantum eigensolver algorithm within the code space.  We show the circuit is fault tolerant, given either a native or fault tolerant implementation of an $e^{-i\theta Z\otimes Z}$ gate.  We find from numerical studies that thresholds exist for these protocols which can range for a $16\times 16$ Fermi-Hubbard model between $99.9\%$ under optimistic assumptions and $99.997\%$ under more pessimistic assumptions.  These numbers are exacting, but it is worth noting that such an error detecting code arises for free when using a Bravyi-Kitaev superfast encoding and thus may be of great value for early fault tolerant calculations where low-distance quantum error correcting codes are possible but high distance codes may be impractical.  In such cases, a modest amount of error correction will push us below such a threshold and then our techniques can be used to further improve the quality of the estimates of a variational eigensolver after post selecting on no-error-detection events.

This work opens a number of interesting questions about the role that fermionic representations may have going forward in quantum computing.  In particular, while our work suggests the existence of a threshold for gate errors in certain error detection schemes, demonstration of large thresholds will be necessary to show an advantage for either Fermi-Hubbard or chemistry simulations on NISQ era quantum devices.  It is possible that with refinements of the methods presented here or other forms of error mitigation that such improvements could be demonstrated.  Further, while the weight of the stabilizers that need to be measured are modest in our setting, they still are likely to prove experimentally challenging on near term hardware with limited connectivity and finding improved methods that require even lower weight checks may prove to be valuable in such settings to enable these methods to be practically deployed.  Finally, while this work has focused on fermionic representations of creation and annihilation operators, a question arises about whether similar schemes could be proposed for systems with different particle statistics such as bosons or anyons. Further it remains an open question whether the notions of fermionic encodings considered here could be generalized to allow different encodings of operators, such as adders, that are not normally thought of in the context of particle encodings.  Finding more general ways to blur the boundaries between quantum algorithms and quantum error correcting codes may not only provide new ways of reducing the costs of quantum error correction but also provide a deeper understanding of the nature and structure of fault tolerant quantum algorithms.

\section*{Acknowledgements}
NW would like to acknowledge funding for this work from Google Inc. This material is based upon work supported by the U.S. Department of Energy, Office of Science, National Quantum Information Science Research Centers, Co-design Center for Quantum Advantage (C2QA) under contract number DE-SC0012704 (PNNL FWP 76274).

\bibliographystyle{unsrt}
\bibliography{references}
\end{document}